\documentclass[a4paper,10pt]{article}
\usepackage{graphicx}
\usepackage{amsmath,amssymb,amsfonts,amsthm}
\usepackage{url, enumerate,anysize}
\usepackage{lscape}
\usepackage[subnum]{cases}

\usepackage[usenames,dvipsnames,svgnames,table]{xcolor}
\usepackage{color}
\usepackage[colorlinks=true,linkcolor=black, citecolor=blue, urlcolor=blue]{hyperref}
\usepackage{graphicx}
\usepackage{tikz}
\usepackage[titletoc]{appendix}
\newtheorem{thm}{Theorem}[section]

\newtheorem{lem}[thm]{Lemma}

\usepackage[usenames,dvipsnames,svgnames,table]{xcolor}

\newenvironment{mylemma}[1]{{ \textbf{\textit{Proof of Lemma #1:}}}}{}
\newenvironment{mythm}[1]{{ \textbf{\textit{Proof of Theorem #1:}}}}{}

\title{Some parametrized dynamic priority policies for 2-class M/G/1 queues: completeness and applications}
\author{Manu K. Gupta$^1$\footnote{Emails: Manu K. Gupta (manu-kumar.gupta@irit.fr), N. Hemachandra (nh@iitb.ac.in) and J. Venkateswaran (jayendran@iitb.ac.in)}, N. Hemachandra$^2$ and J. Venkateswaran$^2$\\ $^1$IRIT, 2 rue C. Camichel, Toulouse, France \\$^2$Industrial Engineering and Operations Research, IIT Bombay}

\marginsize{1.5cm}{1.5cm}{1.5cm}{1.5cm}
\begin{document}
\maketitle 

\begin{abstract}
Completeness of a dynamic priority scheduling scheme is of fundamental importance for the optimal control of queues in areas {as diverse as computer communications, communication networks, supply chains and manufacturing systems. Our \textit{first main contribution} is to {identify the mean waiting time completeness as a unifying aspect} for four different dynamic priority scheduling schemes by proving their completeness and equivalence in 2-class M/G/1 queue.} These dynamic priority schemes are earliest due date based, head of line priority jump, relative priority and probabilistic priority.

In our \textit{second main contribution}, we {characterize} the optimal scheduling policies for the case studies in different domains {by exploiting the completeness of above dynamic priority schemes}. {The major theme of second main contribution is resource allocation/optimal control in revenue management problems for contemporary systems such as cloud computing, high performance computing, etc., where congestion is inherent. {
Using completeness and theoretically tractable nature of relative priority policy, we study the impact of approximation in a fairly generic data network utility framework. We introduce the notion of min-max fairness in multi-class queues and show that a simple global FCFS policy is min-max fair. Next, we re-derive the celebrated $c/\rho$ rule for 2-class M/G/1 queues by an elegant argument and also simplify a complex joint pricing and scheduling problem for a wider class of scheduling policies.  
} }
\end{abstract}


\section{Introduction}

{

In recent years, there has been growing interest in dynamic resource allocation/optimal control in revenue management problems from various research communities such as computer science, management science, communication engineering, etc. In such service systems, congestion is inherent and our focus is on contemporary systems such as cloud computing, high performance computing, data networks, etc. These systems often consist of various types of incoming traffic, seeking differential service requirements. It is a fundamental concern for service providers to allocate suitable network resources to appropriate traffic classes so as to maximize either the resource utilization or obtain maximal revenue or provide better quality of service, etc. 


Multi-class queues offer a flexible way of modeling a variety of complex dynamic real world problems where customers arrive over time for service and service discrimination is a major criterion. Thus, the choice of queue discipline is important. Different types of priority schemes are possible to schedule the customers which are competing for service at a common resource. Absolute or strict priority to one class of customers usually results in starvation of resource for a very long time to the lower priority classes of customers. This motivates the use of dynamic priority scheduling schemes.

}

There are various types of \textit{parametrized} dynamic priority rules to overcome the starvation of lower priority customers in multi-class queues. Kleinrock proposed Delay Dependent Priority (DDP) scheme based on delay in queues (see \cite{Kleinrock1964}). Some other parametrized dynamic priority rules are Earliest Due Date (EDD) based dynamic priority (see \cite{EDDpriority}), Head Of Line Priority Jump (HOL-PJ)~(see \cite{holpj}) and probabilistic priority (see \cite{jiang2002delay}).  Relative priority, recently proposed in \cite{haviv2}, is yet an another class of parametrized dynamic priority scheme which is based on the number of customers in each class. Each dynamic priority scheme has its own applicability and limitations. {One of the central themes of this paper is to provide a unifying aspect for all these priority schemes by identifying them as \textit{complete} {{and by relating them to the completeness of extended DDP}}.} {We now discuss different types of dynamic priority schemes followed by a discussion on the significance of completeness.}

{

EDD dynamic priority scheme often finds application in project scheduling, where multiple tasks (jobs) need to be completed before their respective deadlines using shared resources. Due to parametrized (by urgency numbers or equivalently due dates) nature of EDD priority scheme, appropriate urgency numbers can be designed for each type of job in a given project management problem.} No additional processing delay is involved with HOL-PJ as compared to HOL (see \cite{holpj}). Thus, HOL-PJ is \textit{computationally most efficient} dynamic priority scheme among all dynamic priority schemes discussed above. Note that HOL-PJ will have relatively \textit{less} \textit{switching rate} due to its mechanism being similar to HOL. In probabilistic priority discipline, service is provided to each class based on polling and a pre-defined parameter associated with each class. This scheme associates a compact real valued parameter for each class and does not use the information about number or delay in queue while scheduling the customers. This can be heavily exploited in building simulators for multi-class queues and for solving optimal control problems (see Section \ref{applications} for few examples). One of the major drawbacks with this scheduling discipline is the unavailability of an exact expression for mean waiting time of each class. Relative priority queue discipline overcomes this drawback. Relative priority scheme associates a compact parameter in 2-class queue and exact expressions for mean waiting time are known (see \cite{haviv2}). Hence, relative priority scheme can be used to simplify optimal control problems. Few such examples are discussed in this paper (see Section \ref{cmu_rule} and \ref{joint_pricing}).

Optimal control of multi-class queueing systems has received significant attention due to its applications in computers, communication networks, and manufacturing systems (see \cite{bertsimas1995achievable}, \cite{bertsimas1996conservation}, \cite{hassin2009use} and references therein). One of the main tools for such control problems is to characterize the achievable region for performance measure of interest, and then use optimization methods to find the optimal control policy (see \cite{gupta20152}, \cite{2classpolling} and \cite{li2012delay}). Optimal control policy for certain nonlinear optimization problems for 2-class work conserving queueing systems is derived in \cite{hassin2009use}. A finite step algorithm for optimal pricing and admission control is proposed in \cite{sinha2010pricing} by using a complete class of parametrized (delay dependent) dynamic priority. Optimal control policy in 2-class polling (non work conserving) system for certain optimization problems using achievable region approach is recently developed  in \cite{2classpolling}. In each of these, a suitable class of parametrized dynamic priority schemes are used; to ensure optimality, such classes have to be \textit{complete} as discussed below. 

Average waiting time vectors form a nice geometric structure (polytope) driven by Kleinrock's conservation laws under certain scheduling assumptions for multi-class single server priority queues (see \cite{coffman1980characterization}, \cite{shanthikumar1992multiclass}). This kind of structure also helps if one wants to solve an optimal control problem over a certain set of scheduling policies. Researchers in this field have come up with geometrical structure of achievable region in case of multiple servers and even for some networks (See \cite{federgruen}, \cite{bertsimas}). Unbounded achievable region for mean waiting time in 2-class deterministic polling system is identified in \cite{2classpolling} and a unifying conservation law is recently proposed in \cite{ayesta2007unifying}. {Achievable region for {nonlinear} performance measures have also been explored in literature; for example variance of waiting time in single class queue by \cite{gupta20152} and waiting time tail probability in 2-class queue by \cite{gupta2015conservation}. }

A parametrized scheduling policy is called \textit{complete} by Mitrani and Hine (\cite{complete}) if it achieves all possible vectors of mean waiting times in the achievable region. This question of completeness is important in the following aspect. A complete scheduling class can be used to find the optimal control policy over the set of scheduling disciplines. Discriminatory processor sharing (DPS) class of parametrized dynamic priority is identified as \textit{complete} policy in 2-class M/G/1 queue and used to determine the optimal control policy in \cite{hassin2009use}. This idea of completeness is also useful in designing synthesis algorithms where service provider wants to design a system with certain service level (mean waiting time) for each class. Federgruen and Groenvelt (\cite{federgruen}) devised a synthesis algorithm by using the completeness of mixed dynamic priority which is based on delay dependent priority scheme proposed in \cite{Kleinrock1964}. {This paper provides a unifying presentation for different dynamic priority scheduling schemes by identifying them as complete and solves certain contemporary resource allocation problems in the context of revenue management. We also revisit some classical queuing problems (fairness and $c\mu$ rule) and demonstrate the applicability of these ideas. Thus, the contributions of this paper are two-fold: }
\begin{enumerate}
\item Four different dynamic priority scheduling schemes are identified to be complete. 
\begin{itemize}
\item[-] Explicit closed form expressions for equivalence between these scheduling schemes. 
\item[-] Completeness and equivalence provide a unifying view for different scheduling schemes. 
\end{itemize}
{
\item Applications in solving optimal control problems in different domains:
\begin{itemize}
\item[-] High performance computing, cloud computing, $c/\rho$ rule, a joint pricing and scheduling problem.  
\item[-] Optimal utility in a data network.
\item[-] Min-max fairness nature of global FCFS policy.
\end{itemize}}
\end{enumerate}
We now provide a brief summary and methodology for all the results. We first argue that extended DDP forms a complete class using some of the results in literature. Further completeness of other dynamic priority schemes (EDD, relative, HOL-PJ and probabilistic priority) is established via equivalence with extended DDP in 2-class M/G/1 queue.

{Some appropriate optimal control problems are described in the context of high performance computing facility and cloud computing. We formulate an optimization problem to find a scheduling scheme which maximizes the utility of High Performance Computing (HPC) server in the presence of price sensitive demand. Another optimization problem is formulated to maximize the revenue rate for cloud computing server while ensuing certain quality of service for each class of incoming traffic. These {optimal control problems} exploit the completeness of relative priority discipline. {Further, these completeness results are used to propose a simpler way {of obtaining the} celebrated $c/\rho$ rule for 2-class M/G/1 queues. {A complex} joint pricing and scheduling problem considered in \cite{sinha2010pricing} is simplified using these ideas and we identify that the optimal scheduling scheme obtained in \cite{sinha2010pricing} is indeed optimal for a wider class of scheduling policies.}}

We revisit the problem of obtaining optimal utility in 2-class delay sensitive data network considered in \cite{jiang2002delay}. Approximate utility is obtained for this network by using probabilistic priority scheme (see \cite{jiang2002delay}). The stationary mean waiting time expressions are difficult to derive for probabilistic priority scheme and hence the utility computed in \cite{jiang2002delay} is approximate. We first observe that the probabilistic priority scheme consider in \cite{jiang2002delay} is actually a \textit{complete} scheduling policy. We exploit the completeness of relative priority to obtain {optimal} relative priority parameter that maximizes the network utility. The maximum utility by using the approximate mean waiting time in probabilistic priority scheme is termed as approximate utility. We exploit the theoretical tractability of global FCFS scheduling discipline to compute this approximate utility for suitably chosen system parameters. In such instances, we note that optimal utility can be quite different from the approximate one.

Fairness is an important notion for a scheduler in multi-class queues (see \cite{wierman_pe}, \cite{wierman_survey}). We introduce the notion of minmax fairness in terms of minimizing the maximum dissatisfaction (mean waiting time) of each customer's class in a multi-class queue. {We argue that a simple global FCFS policy is the only solution for this  minmax fairness problem among the set of all non-preemptive, {non-anticipative} and work conserving scheduling policies by exploiting the idea completeness.
}

{
An earlier version of this work is published in \cite{valuetoolsppr} where completeness of EDD and HOL-PJ was described. Applications of these ideas in obtaining $c/\rho$ rule and minmax fairness {were briefly} discussed in \cite{valuetoolsppr}. Proof of $c/\rho$ rule was not described in \cite{valuetoolsppr}. This paper investigates the completeness of relative and probabilistic priority schemes. {Further, in this paper, we use complete classes in finding the optimal scheduling schemes for cloud computing and high performance computing server. We also discuss the applications of these ideas in  obtaining optimal utility in a data network, including the impact of approximate mean waiting time expression under probabilistic priority scheme (see \cite{jiang2002delay}) and in a joint pricing and scheduling problem (see \cite{sinha2010pricing}). 

}
  
}
\subsection{Paper organization}
This paper is organized as follows. Section \ref{sec:description} describes the idea of completeness and four different types of parametrized dynamic priority schemes. Section \ref{sec:completeness_proofs} presents the results on completeness and equivalence between them. Section \ref{applications} discusses applications of these completeness results in solving various optimal control problems. Section \ref{sec:conclusion} ends with discussion and directions for future avenues.

\section{Parametrized dynamic priority policies and their completeness}
\label{sec:description}

In this section, we briefly discuss the {notion} of completeness and different types of parametrized dynamic priority disciplines in a {multi-class single server} M/G/1 queue.

\indent Consider a single server system with $N$ different classes of customers arriving as  independent Poisson streams each with rate $\lambda_i$ and let the mean service time be $1/\mu_i$ for class $i \in \{1,\cdots,N\}$. Let $\rho_i = \lambda_i/\mu_i,~i=1,\cdots,N$ and $\rho = \rho_1 + \rho_2 +\cdots +\rho_N$. Assume that $\rho < 1$, i.e., system attains steady state. Let the service time variance for each class be finite, i.e., $\sigma_i^2< \infty,~i=1,\cdots,N$. The performance of the system is measured by vector $\mathbf{W} = (w_1, w_2, \cdots, w_N)$, where $w_i$ is the expected waiting time of class $i$ jobs in steady state. It is obvious that all performance vectors are not possible; for example $\mathbf{W =0}$ (see \cite{complete}). We restrict our attention to scheduling disciplines which satisfy following conditions:
\begin{enumerate}
\item Service discipline is non-preemptive.
\item Server is not idle when there are jobs in the system (work conserving). 
\item Information about remaining processing time does not affect the system in any way ({non-anticipative}).
\end{enumerate}
Kleinrock's conservation law holds under above scheduling assumptions (see \cite{Kleinrock1965}):
\begin{equation}\label{ConLaw}
\sum_{i=1}^N\rho_i w_i = \dfrac{\rho W_0}{1-\rho}
\end{equation} 
where $W_0 = \sum\limits_{i=1}^N\dfrac{\lambda_i}{2}\left(\sigma_i^2 + \dfrac{1}{\mu_i^2}\right)$. This equation defines a $(N-1)$ dimensional \textit{hyperplane} in $N$ dimensional space of $\mathbf{W}$. 
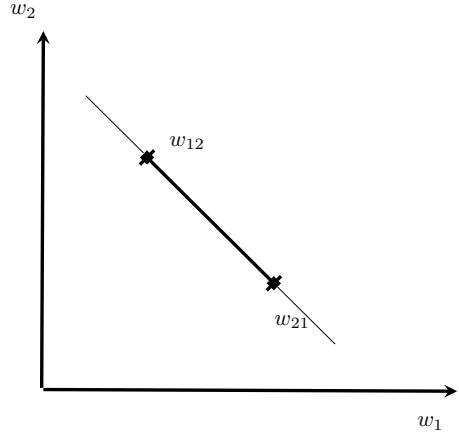
\begin{figure}[htb!]
\centering
\resizebox{0.35 \textwidth}{!}{
\ifx\du\undefined
  \newlength{\du}
\fi
\setlength{\du}{15\unitlength}
\begin{tikzpicture}
\pgftransformxscale{1.000000}
\pgftransformyscale{-1.000000}
\definecolor{dialinecolor}{rgb}{0.000000, 0.000000, 0.000000}
\pgfsetstrokecolor{dialinecolor}
\definecolor{dialinecolor}{rgb}{1.000000, 1.000000, 1.000000}
\pgfsetfillcolor{dialinecolor}
\pgfsetlinewidth{0.100000\du}
\pgfsetdash{}{0pt}
\pgfsetdash{}{0pt}
\pgfsetbuttcap
{
\definecolor{dialinecolor}{rgb}{0.000000, 0.000000, 0.000000}
\pgfsetfillcolor{dialinecolor}
\pgfsetarrowsend{stealth}
\definecolor{dialinecolor}{rgb}{0.000000, 0.000000, 0.000000}
\pgfsetstrokecolor{dialinecolor}
\draw (10.050000\du,15.000000\du)--(22.250000\du,15.050000\du);
}
\pgfsetlinewidth{0.100000\du}
\pgfsetdash{}{0pt}
\pgfsetdash{}{0pt}
\pgfsetbuttcap
{
\definecolor{dialinecolor}{rgb}{0.000000, 0.000000, 0.000000}
\pgfsetfillcolor{dialinecolor}
\pgfsetarrowsend{stealth}
\definecolor{dialinecolor}{rgb}{0.000000, 0.000000, 0.000000}
\pgfsetstrokecolor{dialinecolor}
\draw (10.000000\du,14.950000\du)--(10.050000\du,4.350000\du);
}
\pgfsetlinewidth{0.100000\du}
\pgfsetdash{}{0pt}
\pgfsetdash{}{0pt}
\pgfsetbuttcap
{
\definecolor{dialinecolor}{rgb}{0.000000, 0.000000, 0.000000}
\pgfsetfillcolor{dialinecolor}
}
\definecolor{dialinecolor}{rgb}{0.000000, 0.000000, 0.000000}
\pgfsetstrokecolor{dialinecolor}
\draw (13.176777\du,8.176777\du)--(16.773223\du,11.773223\du);
\pgfsetlinewidth{0.100000\du}
\pgfsetdash{}{0pt}
\pgfsetmiterjoin
\pgfsetbuttcap
\definecolor{dialinecolor}{rgb}{0.000000, 0.000000, 0.000000}
\pgfsetfillcolor{dialinecolor}
\fill (12.893934\du,8.106066\du)--(13.106066\du,7.893934\du)--(13.318198\du,8.106066\du)--(13.106066\du,8.318198\du)--cycle;
\definecolor{dialinecolor}{rgb}{0.000000, 0.000000, 0.000000}
\pgfsetstrokecolor{dialinecolor}
\draw (12.893934\du,8.318198\du)--(13.318198\du,7.893934\du);
\pgfsetlinewidth{0.100000\du}
\pgfsetdash{}{0pt}
\pgfsetmiterjoin
\pgfsetbuttcap
\definecolor{dialinecolor}{rgb}{0.000000, 0.000000, 0.000000}
\pgfsetfillcolor{dialinecolor}
\fill (17.056066\du,11.843934\du)--(16.843934\du,12.056066\du)--(16.631802\du,11.843934\du)--(16.843934\du,11.631802\du)--cycle;
\definecolor{dialinecolor}{rgb}{0.000000, 0.000000, 0.000000}
\pgfsetstrokecolor{dialinecolor}
\draw (17.056066\du,11.631802\du)--(16.631802\du,12.056066\du);
\pgfsetlinewidth{0.010000\du}
\pgfsetdash{}{0pt}
\pgfsetdash{}{0pt}
\pgfsetbuttcap
{
\definecolor{dialinecolor}{rgb}{0.000000, 0.000000, 0.000000}
\pgfsetfillcolor{dialinecolor}
\definecolor{dialinecolor}{rgb}{0.000000, 0.000000, 0.000000}
\pgfsetstrokecolor{dialinecolor}
\draw (16.950000\du,11.950000\du)--(18.650000\du,13.650000\du);
}
\pgfsetlinewidth{0.010000\du}
\pgfsetdash{}{0pt}
\pgfsetdash{}{0pt}
\pgfsetbuttcap
{
\definecolor{dialinecolor}{rgb}{0.000000, 0.000000, 0.000000}
\pgfsetfillcolor{dialinecolor}
\definecolor{dialinecolor}{rgb}{0.000000, 0.000000, 0.000000}
\pgfsetstrokecolor{dialinecolor}
\draw (11.302071\du,6.272071\du)--(13.002071\du,7.972071\du);
}
\definecolor{dialinecolor}{rgb}{0.000000, 0.000000, 0.000000}
\pgfsetstrokecolor{dialinecolor}
\node[anchor=west] at (20.850000\du,16.000000\du){$w_1$};
\definecolor{dialinecolor}{rgb}{0.000000, 0.000000, 0.000000}
\pgfsetstrokecolor{dialinecolor}
\node[anchor=west] at (8.845000\du,3.760000\du){$w_2$};
\definecolor{dialinecolor}{rgb}{0.000000, 0.000000, 0.000000}
\pgfsetstrokecolor{dialinecolor}
\node[anchor=west] at (13.545000\du,7.660000\du){$w_{12}$};
\definecolor{dialinecolor}{rgb}{0.000000, 0.000000, 0.000000}
\pgfsetstrokecolor{dialinecolor}
\node[anchor=west] at (16.640000\du,12.975000\du){$w_{21}$};
\end{tikzpicture}}
\caption{Achievable performance vectors in a 2-class M/G/1 queue \cite{mitranibook}}
\label{2classline}
\end{figure}

In case of two classes, all achievable performance vectors $\mathbf{W} = (w_1,w_2)$ form the points lying on a \textit{straight line segment} defined by Kleinrock's conservation law as shown in Figure \ref{2classline}. There are two special points on this line segment $\mathbf{w_{12}}$ and $\mathbf{w_{21}}$. These two points correspond to the mean waiting time vector when class 1 and class 2 are given strict priority, respectively. The priority policy (1,2) yields the lowest possible average waiting time for type 1 and the highest possible one for type 2; the situation is reversed with the policy (2,1). Thus, no point to the left of (1,2) or to the right of (2,1) can be achieved. Clearly, every point in the line segment is a convex combination of the extreme points $\mathbf{w_{12}}$ and $\mathbf{w_{21}}$.

All achievable performance vectors lie in ($N-1$) dimensional \textit{hyperplane} defined by above conservation law for $N$ classes of customers. There are $(N)!$ extreme points, corresponding to $(N)!$  non-preemptive strict priority policies. Hence, the set of achievable performance vectors form a \textit{polytope} with these \textit{vertices}. Refer Figure \ref{3classline} for polytope corresponding to three classes of customers. Note that it has $(3)! = 6$ vertices.  

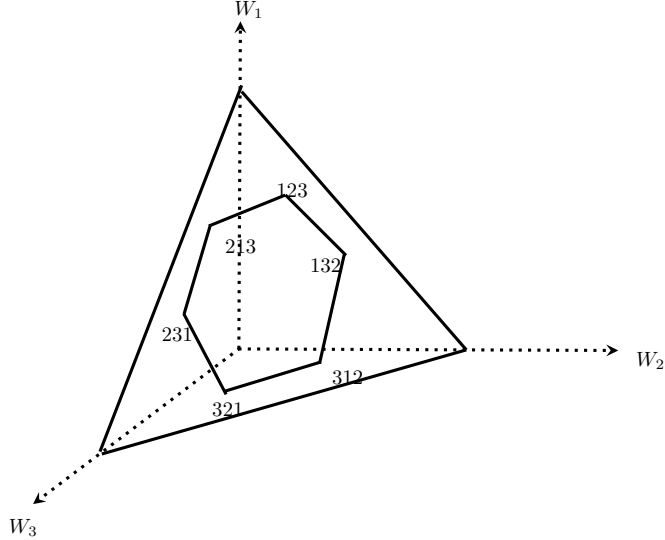
\begin{figure}[htb!]
\centering
\resizebox{0.5 \textwidth}{!}{
\ifx\du\undefined
  \newlength{\du}
\fi
\setlength{\du}{15\unitlength}
\begin{tikzpicture}
\pgftransformxscale{1.000000}
\pgftransformyscale{-1.000000}
\definecolor{dialinecolor}{rgb}{0.000000, 0.000000, 0.000000}
\pgfsetstrokecolor{dialinecolor}
\definecolor{dialinecolor}{rgb}{1.000000, 1.000000, 1.000000}
\pgfsetfillcolor{dialinecolor}
\pgfsetlinewidth{0.100000\du}
\pgfsetdash{{\pgflinewidth}{0.200000\du}}{0cm}
\pgfsetdash{{\pgflinewidth}{0.200000\du}}{0cm}
\pgfsetbuttcap
{
\definecolor{dialinecolor}{rgb}{0.000000, 0.000000, 0.000000}
\pgfsetfillcolor{dialinecolor}
\pgfsetarrowsend{stealth}
\definecolor{dialinecolor}{rgb}{0.000000, 0.000000, 0.000000}
\pgfsetstrokecolor{dialinecolor}
\draw (13.645204\du,16.479308\du)--(25.845204\du,16.529308\du);
}
\pgfsetlinewidth{0.100000\du}
\pgfsetdash{{\pgflinewidth}{0.200000\du}}{0cm}
\pgfsetdash{{\pgflinewidth}{0.200000\du}}{0cm}
\pgfsetbuttcap
{
\definecolor{dialinecolor}{rgb}{0.000000, 0.000000, 0.000000}
\pgfsetfillcolor{dialinecolor}
\pgfsetarrowsend{stealth}
\definecolor{dialinecolor}{rgb}{0.000000, 0.000000, 0.000000}
\pgfsetstrokecolor{dialinecolor}
\draw (13.659685\du,16.476802\du)--(13.709685\du,5.876802\du);
}
\pgfsetlinewidth{0.100000\du}
\pgfsetdash{{\pgflinewidth}{0.200000\du}}{0cm}
\pgfsetdash{{\pgflinewidth}{0.200000\du}}{0cm}
\pgfsetbuttcap
{
\definecolor{dialinecolor}{rgb}{0.000000, 0.000000, 0.000000}
\pgfsetfillcolor{dialinecolor}
\pgfsetarrowsend{stealth}
\definecolor{dialinecolor}{rgb}{0.000000, 0.000000, 0.000000}
\pgfsetstrokecolor{dialinecolor}
\draw (13.700000\du,16.450000\du)--(7.050000\du,21.500000\du);
}
\pgfsetlinewidth{0.100000\du}
\pgfsetdash{}{0pt}
\pgfsetdash{}{0pt}
\pgfsetbuttcap
{
\definecolor{dialinecolor}{rgb}{0.000000, 0.000000, 0.000000}
\pgfsetfillcolor{dialinecolor}
\definecolor{dialinecolor}{rgb}{0.000000, 0.000000, 0.000000}
\pgfsetstrokecolor{dialinecolor}
\draw (13.750000\du,7.950000\du)--(9.200000\du,19.800000\du);
}
\pgfsetlinewidth{0.100000\du}
\pgfsetdash{}{0pt}
\pgfsetdash{}{0pt}
\pgfsetbuttcap
{
\definecolor{dialinecolor}{rgb}{0.000000, 0.000000, 0.000000}
\pgfsetfillcolor{dialinecolor}
\definecolor{dialinecolor}{rgb}{0.000000, 0.000000, 0.000000}
\pgfsetstrokecolor{dialinecolor}
\draw (20.965204\du,16.509308\du)--(9.259305\du,19.879305\du);
}
\pgfsetlinewidth{0.100000\du}
\pgfsetdash{}{0pt}
\pgfsetdash{}{0pt}
\pgfsetbuttcap
{
\definecolor{dialinecolor}{rgb}{0.000000, 0.000000, 0.000000}
\pgfsetfillcolor{dialinecolor}
\definecolor{dialinecolor}{rgb}{0.000000, 0.000000, 0.000000}
\pgfsetstrokecolor{dialinecolor}
\draw (13.750000\du,8.150000\du)--(20.965204\du,16.509308\du);
}
\pgfsetlinewidth{0.100000\du}
\pgfsetdash{}{0pt}
\pgfsetdash{}{0pt}
\pgfsetbuttcap
{
\definecolor{dialinecolor}{rgb}{0.000000, 0.000000, 0.000000}
\pgfsetfillcolor{dialinecolor}
\definecolor{dialinecolor}{rgb}{0.000000, 0.000000, 0.000000}
\pgfsetstrokecolor{dialinecolor}
\draw (12.700000\du,12.500000\du)--(15.150000\du,11.500000\du);
}
\pgfsetlinewidth{0.100000\du}
\pgfsetdash{}{0pt}
\pgfsetdash{}{0pt}
\pgfsetbuttcap
{
\definecolor{dialinecolor}{rgb}{0.000000, 0.000000, 0.000000}
\pgfsetfillcolor{dialinecolor}
\definecolor{dialinecolor}{rgb}{0.000000, 0.000000, 0.000000}
\pgfsetstrokecolor{dialinecolor}
\draw (13.200000\du,17.850000\du)--(16.300000\du,16.900000\du);
}
\pgfsetlinewidth{0.100000\du}
\pgfsetdash{}{0pt}
\pgfsetdash{}{0pt}
\pgfsetbuttcap
{
\definecolor{dialinecolor}{rgb}{0.000000, 0.000000, 0.000000}
\pgfsetfillcolor{dialinecolor}
\definecolor{dialinecolor}{rgb}{0.000000, 0.000000, 0.000000}
\pgfsetstrokecolor{dialinecolor}
\draw (15.155187\du,11.495187\du)--(17.100000\du,13.450000\du);
}
\pgfsetlinewidth{0.100000\du}
\pgfsetdash{}{0pt}
\pgfsetdash{}{0pt}
\pgfsetbuttcap
{
\definecolor{dialinecolor}{rgb}{0.000000, 0.000000, 0.000000}
\pgfsetfillcolor{dialinecolor}
\definecolor{dialinecolor}{rgb}{0.000000, 0.000000, 0.000000}
\pgfsetstrokecolor{dialinecolor}
\draw (16.250000\du,16.950000\du)--(17.050000\du,13.400000\du);
}
\pgfsetlinewidth{0.100000\du}
\pgfsetdash{}{0pt}
\pgfsetdash{}{0pt}
\pgfsetbuttcap
{
\definecolor{dialinecolor}{rgb}{0.000000, 0.000000, 0.000000}
\pgfsetfillcolor{dialinecolor}
\definecolor{dialinecolor}{rgb}{0.000000, 0.000000, 0.000000}
\pgfsetstrokecolor{dialinecolor}
\draw (11.900000\du,15.350000\du)--(12.750000\du,12.450000\du);
}
\pgfsetlinewidth{0.100000\du}
\pgfsetdash{}{0pt}
\pgfsetdash{}{0pt}
\pgfsetbuttcap
{
\definecolor{dialinecolor}{rgb}{0.000000, 0.000000, 0.000000}
\pgfsetfillcolor{dialinecolor}
\definecolor{dialinecolor}{rgb}{0.000000, 0.000000, 0.000000}
\pgfsetstrokecolor{dialinecolor}
\draw (13.250000\du,17.950000\du)--(11.890187\du,15.340187\du);
}
\definecolor{dialinecolor}{rgb}{0.000000, 0.000000, 0.000000}
\pgfsetstrokecolor{dialinecolor}
\node[anchor=west] at (26.150000\du,16.850000\du){$W_2$};
\definecolor{dialinecolor}{rgb}{0.000000, 0.000000, 0.000000}
\pgfsetstrokecolor{dialinecolor}
\node[anchor=west] at (13.250000\du,5.500000\du){$W_1$};
\definecolor{dialinecolor}{rgb}{0.000000, 0.000000, 0.000000}
\pgfsetstrokecolor{dialinecolor}
\node[anchor=west] at (6.000000\du,22.250000\du){$W_3$};
\definecolor{dialinecolor}{rgb}{0.000000, 0.000000, 0.000000}
\pgfsetstrokecolor{dialinecolor}
\node[anchor=west] at (12.535000\du,18.440000\du){321};
\definecolor{dialinecolor}{rgb}{0.000000, 0.000000, 0.000000}
\pgfsetstrokecolor{dialinecolor}
\node[anchor=west] at (16.380000\du,17.405000\du){312};
\definecolor{dialinecolor}{rgb}{0.000000, 0.000000, 0.000000}
\pgfsetstrokecolor{dialinecolor}
\node[anchor=west] at (15.675000\du,13.770000\du){132};
\definecolor{dialinecolor}{rgb}{0.000000, 0.000000, 0.000000}
\pgfsetstrokecolor{dialinecolor}
\node[anchor=west] at (14.600000\du,11.350000\du){123};
\definecolor{dialinecolor}{rgb}{0.000000, 0.000000, 0.000000}
\pgfsetstrokecolor{dialinecolor}
\node[anchor=west] at (12.950000\du,13.150000\du){213};
\definecolor{dialinecolor}{rgb}{0.000000, 0.000000, 0.000000}
\pgfsetstrokecolor{dialinecolor}
\node[anchor=west] at (10.910000\du,16.015000\du){231};
\end{tikzpicture}}
\caption{Achievable performance vectors in a three class M/G/1 queue \cite{mitranibook}}
\label{3classline}
\end{figure}

\indent If the value of performance vector is $W$ for a given scheduling strategy $S$, we say that $S$ achieves $W$. A family of scheduling strategy is called \textbf{\textit{complete}} if it achieves the polytope described above (see \cite{complete}). The set of all scheduling strategies is trivially a complete family; thus one is interested in a subset of all strategies, parametrized suitably, but complete. In this article, we identify four family of parametrized scheduling strategies which are complete for 2-class M/G/1 queue. We now describe different types of parametrized dynamic priority schemes from literature. Completeness and equivalence of these dynamic priority schemes are discussed in Section \ref{sec:completeness_proofs}. 
\subsection{Delay dependent priority (DDP) policy} 
Delay dependent priority scheme was first introduced by Kleinrock \cite{Kleinrock1964}. The logic of this discipline is as follows. Each customer class is assigned a queue discipline management parameter, $b_i$, $i \in \{1, \cdots, N\}$, $0 \le b_1 \le b_2 \le \cdots \le b_N$. Higher the value of $b_i$, higher is the rate of gaining priority for class $i$ as discussed below. The instantaneous dynamic priority for a customer of class $i$ at time $t$, $q_i(t)$, is given by: 
\begin{equation}
q_i(t) = (t-\tau)\times b_i,~ i = 1,2, \cdots, N,
\label{ddpinst}
\end{equation}
where $\tau$ is the arrival time of the customer. After the current customer is served, the server will pick the customer with the highest instantaneous dynamic priority parameter $q_i(t)$ for service. Ties are broken using First-Come-First-Served rule.

Mean waiting time for $k^{th}$ class under this discipline, $E(W_k^{DDP})$ is given by following recursion \cite{Kleinrock1964}:
\begin{equation}\label{eqn:DDP_recursion}
E(W_k^{DDP}) = \dfrac{\dfrac{W_0}{1-\rho} - \displaystyle\sum_{i=1}^{k-1} \rho_i E(W_i^{DDP})\left(1-\dfrac{b_i}{b_k}\right)}{1-\displaystyle\sum_{i=k+1}^{N}\rho_i\left(1-\dfrac{b_k}{b_i}\right)}
\end{equation}
where $\rho_i = \frac{\lambda_i}{\mu_i}$, $\rho = \sum\limits_{i=1}^N\rho_i$ and $W_0 = \sum\limits_{i=1}^N\frac{\lambda_i}{2}\left(\sigma_i^2 + \frac{1}{\mu_i^2} \right)$ and $0 < \rho <  1$.

Federgruen and Groenevelt \cite{federgruen} proposed a synthesis algorithm by exploiting the completeness of mixed dynamic priority which is based on delay dependent priority. In case of two classes, mean waiting time expressions get simplified for delay dependent priority scheme. {Extended delay dependent priority for 2-class queues is described in Appendix \ref{proof:DDP_cmplt} {which turns out to be \textit{complete}}.}


\subsection{Earliest due date (EDD) dynamic priority policy}  
This parametrized dynamic priority scheme was first proposed in \cite{EDDpriority}. Consider the system setting similar to delay dependent priority scheme with $N$ classes and single server queueing system. Each class $i \in \{1, \cdots, N\}$ has a constant urgency number $u_i$ (weights) associated with it. Without loss of generality, classes are numbered such that $u_1 \leq u_2 \leq \dots \leq u_N$. When a customer from class $i$ arrives at the system at time $t_i$, customer is assigned a real number $t_i + u_i$. The server chooses the next customer to go into service, from those present in queue, as the one with minimum value of $\{t_i + u_i\}$. 
Let $W_{k}^{EDD}$ denote the waiting time of a class $k$ jobs in non pre-emptive priority under this discipline. In steady state, $E(W_{k}^{EDD})$ is given by \cite{EDDpriority}:
  \begin{eqnarray}\label{eqn:EDD_recursion}
  E(W_{k}^{EDD}) = E(W) + \sum_{i=1}^{k-1}\rho_i\int_{0}^{u_k-u_i}P(W_{k}^{EDD} > t)dt 
  - \sum_{i=k+1}^{N}\rho_i\int_{0}^{u_i-u_k}P(W_{i}^{EDD} > t)dt 
  \end{eqnarray}
for $k = 1, \cdots, N$. Here $E(W) = \frac{W_0}{(1-\rho)}$ and $\rho_i$ is the traffic due to class $i \in \{1,\cdots,N\}$.

The formulation of the scheduling discipline in terms of urgency numbers facilitates various interpretations of the model. One primary interpretation of urgency numbers $u_i$'s correspond to the interval until the due date is reached. This model leads to a unified theory of scheduling with earliest due dates, which is an area of great practical importance (see \cite{EDDpriority}).

\subsection{Relative priority policy}

This is another type of dynamic priority scheme proposed in \cite{haviv2}. In this multi-class priority system, a \textit{positive} parameter $p_i$ is associated with each class $i \in \{1,\cdots,N\}$. If there are $n_j$ jobs of class $j$ on service completion, the next job to commence service is from class $i$ with following probability:
\begin{equation}\nonumber
\dfrac{n_i p_i}{\sum\limits_{j=1}^N n_j p_j}, ~~~1 \leq i \leq N
\end{equation}
Mean waiting time for class $k$ customers under this scheduling scheme, $E(W_k^{RP})$, is given by following recursion \cite{haviv2}:
\begin{equation}\label{eqn:RP_recursion}
E(W_k^{RP}) = W_0 + \sum_{j=1}^NE(W_j^{RP}) \rho_j\dfrac{p_j}{p_k + p_j} + \tau_k E(W_k^{RP}),~~~1\leq k \leq N.
\end{equation}  
  where $\tau_k = \displaystyle \sum_{j = 1}^N\rho_j\dfrac{p_j}{p_k + p_j},~~1 \leq k \leq N $.

\subsection{Head of line priority jump (HOL-PJ) policy}

This is another type of parametrized dynamic priority policy proposed in \cite{holpj}. The fundamental principle of HOL-PJ is to give priority to the customers having the largest queueing delay in excess of its delay requirement. In HOL-PJ, an explicit priority is assigned to each class; the more stringent the delay requirement of the class, the higher the priority. From the server's point of view, HOL-PJ is the same as head of line (HOL) strict priority queue. Unlike HOL, the priorities of customers increase as their queueing delay increases relative to their delay requirements. This is performed by customer \textit{priority jumping} (PJ) mechanism (see Figure \ref{PriortyJump}).
\begin{figure}[htb!]
\centering
\includegraphics[scale=0.4]{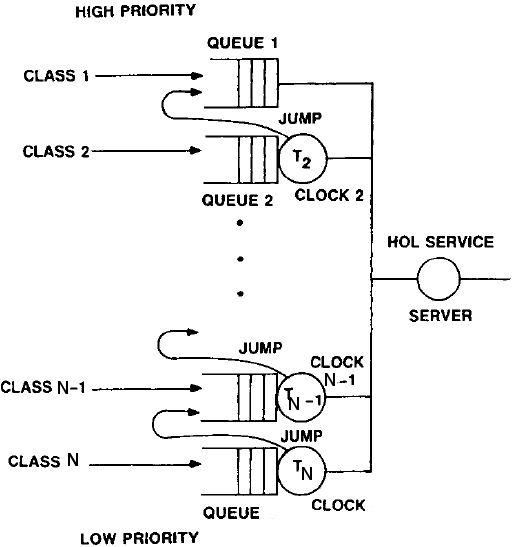}
\caption{Head-of-line with priority jump \cite{holpj}}
\label{PriortyJump}
\end{figure}

Consider a single server serving $N$ classes of customers. Let $D_k,~k=1,2,\ldots,N$ be the delay requirement for class $k$ customers where $0<D_1<D_2<\cdots<D_N\leq \infty$. Class 1 has the most stringent delay requirement and class $N$ the least; class 1 has the highest priority and class $N$ the least. $T_k,~k=2,3,\cdots, N$ is set to $D_k - D_{k-1}$. If a customer is still in queue after a period of time $T_k$, it jumps to the tail of queue $k-1$. Figure \ref{PriortyJump} illustrates the operation of HOL-PJ. Excessive delay of a customer is defined as its queueing delay in excess of its original delay requirement. It is concluded in \cite{holpj} that all the customers are queued according to largeness of their excessive delay. Mean waiting time for class $k$ customers in HOL-PJ, $E(W_k^{HOL-PJ})$, queueing discipline is derived in \cite{holpj}:
\begin{equation}
E(W_k^{HOL-PJ}) = E(W) - \sum_{j =k+1}^N \rho_j \int_0^{\sum_{l = k+1}^jT_l}P(W_j^{HOL-PJ} > t)dt + \sum_{j=1}^{k-1}\rho_j\int_{0}^{\sum_{l = j+1}^k T_l}P(W_k^{HOL-PJ} > t)dt
\end{equation}
Since $T_k = D_k - D_{k-1}$, this gives 
\begin{equation}\label{eqn:holpj_recursion}
E(W_k^{HOL-PJ}) = E(W)  - \sum_{j =k+1}^N \rho_j \int_0^{D_j-D_k}P(W_j^{HOL-PJ} > t)dt + \sum_{j=1}^{k-1}\rho_j\int_{0}^{D_k - D_j}P(W_k^{HOL-PJ} > t)dt.
\end{equation}
Here $E(W)$ is $\frac{W_0}{(1-\rho)}$. Note that above recursion is again not a closed form equation; however, these expressions are useful in deriving mean completeness and equivalence results in Section \ref{sec:completeness_proofs}. 

We briefly describe the practical significance of this model as pointed out in \cite{holpj}. This model can be used in an integrated packet switching node serving multiple classes of delay sensitive traffic (eg., voice and video traffic). Implementation of this discipline is relatively simple and the processing overhead is minimal from server's perspective as its mechanism is similar to head of line strict priority.

\subsection{Probabilistic priority (PP) policy}\label{PP}
This is yet another type of dynamic priority scheme first proposed in \cite{jiang2002delay}. This policy works as follows. Let there be $N$ classes of customers where customers with a smaller class number have higher priority than with the larger class number. PP discipline is non pre-emptive. Each class of customers has its own queue and buffer capacity of the queue is infinite. Customers in the same queue are served in FCFS fashion. Queue $i$ is assigned a parameter $0 \le p_i\le 1$, $i=1,2,\cdots,N.$

At each service completion, the server first polls queue 1 and then polling continues for subsequent queues. If queue $i~(< N)$ and all other queues $j(\neq i)$ are non empty when queue $i$ is polled, the customer at the head of queue $i$ will be served with probability $p_i$. The server polls the next queue $i+1$ with probability $1-p_i$. If some queues are empty when queue $i$ is polled, the head customer of queue $i$ will be served with probability $\hat{p}_i$, the server polls the next non empty busy queue (BQ) with probability $1-\hat{p}_i$. Here $\hat{p}_i$ ($i \in$ BQ) is determined such that the wasted server share of these empty queues is allocated to those non empty queues based on their assigned parameters. Such scheduling discipline is analysed in \cite{jiang2002delay} with a restriction to two class case, thus for two class queue $\hat{p}_1 = p_1 \text{ or } 1$.

If queue $i$ is empty at the time being polled, it will not be served and server polls the next queue $i+1$ with probability 1. If queue $i$ is non empty and at the time being polled but all next queues $j(>i)$ are empty, it will be served with probability 1 instead of $p_i$. This process then repeats at queue $i+1$ which has parameter $p_{i+1}$. In addition, $p_N$ is always set to be one as queue $N$ is the last queue that may be served in a service cycle. Server starts polling queue 1 after each service completion. 

The service cycle refers to the cycle that the server polls queues, serves a customer and restarts polling from queue 1. In each service cycle, one and only one customer is served if the system is not idle. PP discipline is work conserving. Jiang et. al. \cite{jiang2002delay} derived approximate mean waiting time for probabilistic priority scheduling in 2-class queue. Jiang et. al. \cite{jiang2002delay} also derived certain other properties of mean waiting time which are useful in establishing the completeness of this dynamic priority scheme in Section \ref{sec:completeness_proofs}.

Now, we present completeness and equivalence of different dynamic priority schemes for 2-class queue in subsequent section.


\section{Equivalence and completeness of different parametrized policies}
\label{sec:completeness_proofs}

In this section, we prove the completeness of EDD, relative, HOL-PJ and PP parametrized dynamic priority schemes for 2-class M/G/1 queue. {DDP spans the interior of achievable region from \cite{federgruen}. Thus, in 2-class queue, DDP spans the entire achievable region except the two end points. The two end points are achievable by an extended DDP discussed in Appendix \ref{proof:DDP_cmplt}. This implies that extended DDP is complete for two classes. We obtain an explicit one-to-one {nonlinear} transformation from extended DDP class to EDD and to RP. Hence, completeness of EDD and relative priority follows via this equivalence. Completeness of HOL-PJ is argued by identifying the similarity in recursion of mean waiting time for EDD and HOL-PJ. Probabilistic priority scheme is identified as a complete class by exploiting certain properties of mean waiting time. An independent proof of completeness of different dynamic priority schemes is also presented without using the completeness of extended DDP.}

\subsection{EDD based dynamic priority policy}
In case of two classes, the expected waiting time is \cite[Theorem 2]{EDDpriority}:
\begin{eqnarray}\label{eqn:2clsedd1}
E(W_{h}^{EDD}) = E(W) - \rho_l \int_0^u P(T_h[W] > y)dy\\
E(W_l^{EDD}) = E(W) + \rho_h \int_0^u P(T_h[W] > y)dy\label{eqn:2clsedd2}
\end{eqnarray}
Here index $l$ and $h$ are for lower and higher priority class. $u_l$ and $u_h$ are the weights associated with lower and higher classes respectively, $u = u_l - u_h \geq 0$ as $u_l \ge u_h \ge 0$. Let $W(t)$ be the total uncompleted service time of all customers present in the system at time $t$, regardless of class. $W(t) \rightarrow W$ as $t \rightarrow \infty $.
$$T_h[W(t)] = \inf\{t^{'} \geq  0 ;~ \hat{W}_h(t+t^{'}: W(t)) = 0\}$$
where $\hat{W}_h(t+t^{'}: W(t))$ is the workload of the server at time $t+t^{'}$ given an initial workload of $W(t)$ at time $t$ and considering the input workload from class $h$ only after time $t$.\\
\indent Consider the more general setting (in the view of completeness) with this type of priority where $u_1,~ u_2 \geq 0$ be the weights associated with class 1 and class 2. Let $\bar{u} = u_1 - u_2$. Thus $\bar{u}$ can take any value in the extended real line $[ -\infty, \infty]$. Class 1 will have higher or lower priority depending on $\bar{u}$ being negative or positive. By using equations (\ref{eqn:2clsedd1}) and (\ref{eqn:2clsedd2}), mean waiting time for this general setting in case of two classes can be written as:
 \begin{eqnarray}\label{eqn:EDDcombined1}
 E(W_1^{EDD}) &=& E(W) + \rho_2\left[\int_0^{\bar{u}}P(T_2(W) > y)dy ~\mathbf{1}_{\{\bar{u} \geq 0\}}\right.
 \left.-\int_0^{-\bar{u}}P(T_1(W) > y)dy ~\mathbf{1}_{\{\bar{u} < 0\}} \right]\\\label{eqn:EDDcombined2}
  E(W_2^{EDD}) &=& E(W) + \rho_1\left[\int_0^{\bar{u}}P(T_2(W) > y)dy ~\mathbf{1}_{\{\bar{u} \geq 0\}}\right.
 \left. -\int_0^{-\bar{u}}P(T_1(W) > y)dy ~\mathbf{1}_{\{\bar{u} < 0\}} \right]
 \end{eqnarray}
 
  Note that $\bar{u} = -\infty $ and $\bar{u} = \infty$ result in the corresponding mean waiting times when strict higher priority is given to class 1 and class 2 respectively. Hence, we suspect a one-to-one transformation from DDP to EDD priority policy:
\begin{lem}\label{clm:equivalenceDDPnEDD}
\textit{Delay dependent priority policy and earliest due date priority policy are equivalent in 2-class queues and their priority parameters ($\beta$ and $\bar{u}$) are related as:}{
\begin{eqnarray}\nonumber
\beta =  \frac{W_0 - (1-\rho_1)(1-\rho)\tilde{I}(\bar{u})}{W_0 + \rho_1(1-\rho)\tilde{I}(\bar{u})}\times \mathbf{1}_{\{-\infty \le  \bar{u} < 0\}}
 + ~~ \frac{W_0 + \rho_2(1-\rho)I(\bar{u})}{W_0 - (1-\rho_2)(1-\rho)I(\bar{u})}\mathbf{1}_{\{0 \le  \bar{u} \le \infty\}}
\end{eqnarray}}
where integrals $\tilde{I}(\bar{u})=\int_0^{-\bar{u}}P(T_1(W)>y)dy$ and $I(\bar{u})=\int_0^{\bar{u}}P(T_2(W)>y)dy$. 
\end{lem}
\begin{proof}
See Appendix \ref{proof:lemmaclaim}.
\end{proof}
Note that $\beta$ is a monotone function of $\tilde{I}(\bar{u})$, and $\tilde{I}(\bar{u})$ is a monotone function of $\bar{u}$. Hence by the property of monotonicity, there is a one-to-one transformation between $\bar{u}$ and $\beta$. Since extended DDP is a complete dynamic priority discipline in case of two classes, EDD will also be complete. Thus, we have following result:

\begin{thm}\label{clm:EDDcomplete}
\textit{EDD dynamic priority policy is complete in 2-class queues.}
\end{thm}
An independent proof of above theorem without exploiting the completeness of extended DDP can be seen in Appendix~\ref{proof:extra}.
\subsection{Relative dynamic priority policy}
In case of two classes, mean waiting time is given by (see \cite{haviv2}):
\begin{equation}\label{eqn:2class_relative}
E(W_i^{RP}) = \dfrac{1-\rho p_i}{(1- \rho_1 - p_2 \rho_2)(1 - \rho_2 - p_1\rho_1)-p_1 p_2 \rho_1\rho_2}W_0, ~~i=1,2
\end{equation}   
  where $\rho = \rho_1 + \rho_2$ and $p_1 + p_2 = 1$. Note that $p_1 =1$ and $p_2 = 1$ result in corresponding mean waiting times when strict higher priority is given to class 1 and class 2 respectively. Hence, we can expect a one-to-one transformation from DDP to relative priority priority. We find such an explicit {nonlinear} transformation below.

\begin{lem}\label{clm:RPequiv}
\textit{Delay dependent priority policy and relative priority policy are equivalent in two classes and priority parameters ($\beta$ and $p_1$) are related as:}{
\begin{equation}
\beta = \frac{\mu-\lambda p_1}{(2\mu-\lambda)p_1}\mathbf{1}_{\{0 \le  p_1 < \frac{1}{2}\}} + \frac{(2\mu-\lambda)(1-p_1)}{\mu-\lambda(1-p_1)}\mathbf{1}_{\{\frac{1}{2} \le  p_1 \le 1\}}
\end{equation}}
\end{lem}
\begin{proof}
See Appendix~\ref{proof:lemmaclaim}. 
\end{proof}
Above lemma gives one-to-one transformation between $p_1$ and $\beta$. Since extended DDP is a complete dynamic priority discipline in case of two classes, relative priority will also be complete for two classes of customers:

\begin{thm}\label{clm:RPcomplete}
\textit{Relative dynamic priority scheme is complete in 2-class queues.}
\end{thm}  
An independent proof of above theorem without exploiting the completeness of extended DDP is also given in Appendix~\ref{proof:extra}. 

\subsection{{HOL-PJ dynamic priority policy}}
It can be observed from Equation (\ref{eqn:EDD_recursion}) and (\ref{eqn:holpj_recursion}) that the mean waiting time recursion in HOL-PJ is same as that in EDD priority policy. Urgency number and overdue in EDD correspond to delay requirement and excessive delay in HOL-PJ. Similar to EDD, we consider the more general setting in HOL-PJ where $D_1,D_2 \geq 0$ is the delay requirement associated with class 1 and class 2. Let $\bar{D} = D_1- D_2$ be the parameter associated with HOL-PJ similar to $\bar{u}$ in EDD. Hence, we have the following theorem from our previous result on equivalence of EDD and DDP and completeness of EDD dynamic priority for 2-class queues. 
\begin{thm}
\textit{{There is a one-to-one nonlinear transformation for  {the mean waiting time vector of} HOL-PJ and extended DDP, and hence HOL-PJ is complete for 2-class M/G/1 queues.}} 
\end{thm}

\subsection{Probabilistic priority policy}\label{2classPP}
Approximate mean waiting time expressions under probabilistic priority (PP) scheduling are derived in \cite{jiang2002delay} for two classes of customers. Customers arrive according to independent Poisson processes with rates $\lambda_1$ and $\lambda_2$ for class 1 and class 2 respectively. The service times of class $i$ customers are independent, identically distributed, stochastic variables which have general distribution with finite first and second moments $s_i$ and $s_i^{(2)}$, $i = 1,2$. Let $\bar{i}$ denote the class other than class $i$, i.e., if $i=1,2$ then $\bar{i} = 2,1$ respectively. In each service cycle, $\omega_i$ denotes the probability that the head of the line customer from class $i$ is served when queue $\bar{i}$ is non empty. From definition of $\omega_i$, we have
$$\omega_1 = p_1,~~\omega_2 =1- p_1 \text{ and }\omega_i + \omega_{\bar{i}} = 1$$
where notation $p_i$ is as discussed in Section \ref{PP}. Let $\bar{W}_i$ be the mean waiting time in queue for class $i$, $i=1,~2$. Bounds for average waiting time are derived in \cite{jiang2002delay} along with following results which are useful in exploring completeness of PP scheduling discipline. 
\begin{enumerate}
\item For $\omega_i \in [0,1],~i=1,2$, $\bar{W}_i$ is continuous and monotonically decreasing; $\bar{W}_{\bar{i}}$ is continuous and monotonically increasing.  
\item For a given value $\bar{W}^*_i \in [\bar{W}_i^{'}, \bar{W}_i^{''}],~i=1,2$, there must exist $\omega_i^* \in [0,1]$ such that when $\omega_i = \omega_i^*,~\bar{W}_i = \bar{W}^*_i$ where $\bar{W}_i^{'}$ and $\bar{W}_i^{''}$ are the average waiting times when $\omega_i=1$ and $\omega_i = 0$ respectively.
\end{enumerate}
Note that in case of two classes $p_2$ is always 1 and $\omega_1 = 1$ implies $p_1 = 1$. It is clear from the mechanism of PP queue discipline that $p_1=p_2 = 1$ implies class 1 is given strict priority over class 2. Similarly, $\omega_1 = 0$ implies $p_1 = 0$. This will correspond to class 2 having strict priority over class 1. Hence extreme points of line segment in Figure \ref{2classline} are achievable. Any point on the line segment is achievable from above two results by continuously varying $\omega_1$ in range $(0,1)$. Thus, the following result holds. 

\begin{thm}\label{PP_completeness}
Probabilistic priority queue discipline is \textit{complete} for 2-class $M/G/1$ queues.
\end{thm}

Exact transformations to other priority schemes are not tractable as only approximate mean waiting times are known for the $PP$ priority scheme. 
{
\section{{{Various applications of complete policies}}}\label{applications}}
{In this section, we solve some relevant optimal control problems by exploiting the completeness of different dynamic priority schemes introduced in Section \ref{sec:description}. }

\subsection{Optimal scheduling schemes}
In this section, we use the idea of completeness to obtain the optimal scheduling policy for high performance computing facility and cloud computing systems {in Sections \ref{hpc_facility} and \ref{cloud_computing} respectively}. Also, we recover the optimality of celebrated $c/\rho$ rule (see \cite{mitranibook}, \cite{yao2002dynamic}) for 2-class M/G/1 queue by an elegant argument {in Section \ref{cmu_rule}. Further, a complex joint pricing and scheduling problem is simplified for a wider (the set of all non-preemptive, {non-anticipative} and work conserving) class of scheduling policies in Section \ref{joint_pricing}. }

\subsubsection{{Utility maximization in high performance computing facility}}\label{hpc_facility}
We consider the problem of finding the scheduling policy which maximizes the utility for High Performance Computing (HPC) facility. HPC facilities provide high-speed and large-scale computer processing platforms. The computing power of this high-end technology being scarce, jobs are queued up and will be completed eventually. The utility maximization of such an expensive queuing resource is hence desirable. There is a market of users who are willing to pay a higher usage charge to obtain lower mean waiting time for their jobs. We consider the problem of utility maximization for such a HPC center by  casting it as priority based resource allocation in multi-class queue to achieve differential service.
{

We now provide a specific example of an HPC system {which is being operated as above}. The National Renewable Energy Laboratory (NREL) HPC system is one of the largest HPC systems in the world dedicated to advancing renewable energy and energy efficiency technologies \cite{Users}.  Users are charged certain price for using this facility. However, users can reduce their queue waiting time by paying more to HPC facility. Jobs are given (non-preemptive strict) priority if they pay twice the normal rate \cite{Queues}. The results of this section provide the revenue optimal scheduling scheme for NREL type of HPC systems. }

 Let $\lambda_R$ be the arrival rate for regular jobs and $\lambda_P$ be arrival rate for prime jobs (customers) who can pay higher price for faster service. Assume that $\lambda_P$ and $\lambda_R$ are fixed and follow independent Poisson processes\footnote{This is a standard assumption on arrival processes.} and let the service time be general with finite second moment. Let the stationary mean waiting time for prime and regular class be $E(W_P^\pi)$ and $E(W_R^\pi)$ respectively, for a scheduling policy   $\pi\in \mathcal{F}$ where $\mathcal{F}$ is the set of all non-preemptive, non-anticipative and work conserving policies. Further, assume that the price ($\theta$) for prime class is linearly dependent on $E(W_P^\pi)$ for that class:
$$\theta = a-bE(W_P^\pi)$$
where $a$ and $b$ are the (positive) sensitivity constants driven by market. Note that above pricing scheme is natural and captures the fact that one has to pay higher price to reduce the stationary mean waiting time. The utility for the HPC facility under the given scheduling scheme $\pi$: 
$$U^\pi := w_1 (\theta\lambda_P) + w_2 (E(W_R^\pi))$${
where $w_1$ and $w_2$ are the given weights associated with revenue from prime class and service level for regular class respectively. Note that each component in above utility function depends on the scheduling scheme $\pi$. The objective is to find a scheduling scheme that maximizes the utility among the set of all non-preemptive, non-anticipative and work conserving scheduling disciplines, $\mathcal{F}$. Mathematically, 
\begin{equation}\label{utility_maximization}
\max_{\pi \in \mathcal{F}}~~U^\pi
\end{equation}
By using completeness of relative priority from Theorem \ref{clm:RPcomplete}, utility maximization problem simplifies to:
$$\max_{0 \le p \le 1}~~ w_1 (\theta\lambda_P) + w_2 (E(W_R^p))$${
Note that the above problem is theoretically tractable as compared to the problem (\ref{utility_maximization}) and can be solved by the optimization methods involving second degree polynomials\footnote{The denominator of mean waiting time expression is of second degree in $p$ under the relative priority scheduling scheme (see Equation (\ref{eqn:2class_relative})).}.} Alternatively, consider the following revenue maximization problem with a guaranteed service level constraint on regular type of customers:
$$\max_{\pi \in \mathcal{F}}~~  \theta\lambda_P$$
\hspace{6cm} subject to
\hspace{3cm} $$E(W_R^\pi) \le S_R$$
for a given service level threshold $S_R$ for regular jobs. 
{Again, by invoking Theorem \ref{clm:RPcomplete}, one can achieve the theoretical tractability similar to problem (\ref{utility_maximization}).}
\subsubsection{{Revenue rate maximization in cloud computing}}\label{cloud_computing}
Broadly speaking, cloud computing is the delivery of on-demand computing resources over the internet. It provides the capability through which typically real-time scalable resources such as files, programs, data, hardware, computing, and the third party services can be {accessed} via the network to users. With the cloud, users can access the information technology (IT) resources at any time and from multiple locations, track their usage levels, and scale up their service delivery capacity as needed, without large upfront investments in software or hardware.
Pricing schemes are emerging as an attractive alternative to cope with unused capacities and uncertain demand patterns in the context of cloud computing (see \cite{agmon2013deconstructing}, \cite{borkar2017index}). { In today's world, the fundamental metrics for data centers (cloud computing) are throughput, transactional response time (delay) and the cost \cite{Metrics}. The cloud computing facility has often been modeled as multi-class queuing systems in the literature (see \cite{guo2014dynamic}). }

{
We model a cloud computing server for 2-classes of incoming traffic which is delay as well as price sensitive.} Each class of traffic has a certain Service Level Agreement (SLA) in terms of stationary mean waiting time (delay). Such SLA can also be viewed as deadline for each type of jobs. Cloud computing service provider would maximize the {revenue rate} generated by the throughput of the system. {We devise a method for obtaining optimal scheduling scheme for cloud computing server by formulating an appropriate optimization problem.}

{
Consider the two separate classes of incoming traffic into the system according to independent Poisson arrivals with arrival rates $\lambda_1$ and $\lambda_2$ for class 1 and class 2 respectively. The cloud computing server serves the jobs with service rate $\mu$ with an independent general distribution. {Let $E(W_i^\pi)$ be the stationary mean waiting time for class $i$, $i\in \{1,2\}$, for a scheduling policy  $\pi\in \mathcal{F}$.} Note that the throughput of class 1 and class 2 would be exactly same as the departure rate for class 1 and class 2 respectively. Further, the arrival rates and departure rates are same for a stable queue. Thus, the departure rate (or throughput) for class 1 and class 2 would be $\lambda_1$ and $\lambda_2$ respectively. Throughput generates revenue for the system. Let $\theta_1$ and $\theta_2$ be the price charged for the incoming traffic of class 1 and class 2 respectively. Thus, the total revenue rate will be $\theta_1 \lambda_1 + \theta_2 \lambda_2$. We assume that the incoming traffic to each class is linearly sensitive to the price and stationary mean waiting time  ($E(W_i^\pi)$): 
$$\lambda_i = a_i-b_i\theta_i -c_iE(W_i^\pi)~\text{ for }i=1,~2,$$
where $a_i,~b_i$ and $c_i$ are the (positive) sensitivity constants driven by market with threshold service level agreement,  $T_i$, $i\in \{1,2\}$, for class $i$ traffic. Now, consider the problem of maximizing the revenue rate for cloud computing service provider with SLA constraint for each class of incoming traffic over the scheduling policies $\pi \in \mathcal{F}$. Mathematically, 
$$\max_{\pi \in \mathcal{F}} ~\theta_1 \lambda_1 + \theta_2 \lambda_2$$
\hspace{5cm} subject to: 
$$E(W_1^\pi) \le T_1,~E(W_2^\pi) \le T_2,$$
where $\mathcal{F}$ is a set of all non-preemptive, non-anticipative and work conserving policies. Note that each of the constraint and the objective function in above optimization problem depends on scheduling policy $\pi\in \mathcal{F}$. By using completeness of relative priority from Theorem \ref{clm:RPcomplete}, the above problem simplifies to:
$$\max_{0\le p \le 1} ~\theta_1 \lambda_1 + \theta_2 \lambda_2$$
\hspace{5cm} subject to: 
$$E(W_1^p) \le T_1, ~E(W_2^p) \le T_2,$$
 which is theoretically tractable and can be {solved by the optimization methods involving second degree polynomials\footnote{The denominator of mean waiting time expression is of second degree in $p$ under the relative priority scheduling scheme (see Equation (\ref{eqn:2class_relative})).}.} }

\subsubsection{Optimality of $c/\rho$ rule in 2-class M/G/1 queues}\label{cmu_rule}
{
It is well known in literature (see \cite{mitranibook}, \cite{yao2002dynamic}) that a linear weighted combination of mean waiting time under policy $\pi$, $C^\pi:=\sum\limits_{i=1}^{N}c_i E(W_i^\pi)$, is minimized by $c/\rho$ rule when $\pi\in\mathcal{F}$.} Here $c_i$ and $W_i^\pi$, are the cost and mean waiting time associated with class $i,$ $i\in\{1,\cdots,N\}$, under policy $\pi \in \mathcal{F}$ respectively. This rule states that the optimal scheduling discipline with respect to objective $C^\pi$ is a strict priority scheme where priority is assigned in the {decreasing order} of ratios $c_i/\rho_i$.  

We give the idea for the proof of this result in 2-class M/G/1 queue by exploiting completeness results discussed in this paper. Consider problem \textbf{P1}: 
$$\mathbf{P1}~~\min_{\pi \in \mathcal{F}}~~c_1 E(W_{1}^{\pi})+c_2E(W_{2}^{\pi})$$
Note that optimizing over $\mathcal{F}$ is same as optimizing over set of relative priority by completeness property (see Theorem \ref{clm:RPcomplete}). Thus, \textbf{P1} is equivalent to following transformed problem \textbf{T1}:
 $$\mathbf{T1}~~\min_{ p \in [0, 1] }~~c_1 E(W_{1}^{p})+c_2 E(W_{2}^{p}) $$
 Above optimization problem \textbf{T1} can be easily solved to yield the optimal $c/\rho$ rule (see Appendix \ref{cmurule}).
}

\subsubsection{Joint pricing and scheduling problem}\label{joint_pricing}
The pricing model introduced in \cite{sinha2010pricing} solves a generic problem of pricing surplus server capacity of a stable M/G/1 queue for new (secondary) class of customers without affecting the service level of its existing (primary) customers. Inclusion of secondary customers increases the load and affects the service level of primary customers. Hence, admission control and appropriate scheduling of customers across classes is necessary. This queueing model used both admission control (by pricing and service level) and choice of queue discipline parameter for quality of service discrimination. The objective of the model is to solve joint pricing and scheduling problem such that resource owner's revenue will be maximized while maintaining the promised quality of service level for primary customers; it can be noted that these optimal decision variables can be interpreted as a unique Nash equilibrium of a suitably defined two person non-zero sum game where strategy sets of each player depend on the strategy used by another player \cite{NEremark}. {This pricing model under the preemptive scheduling scheme is solved in \cite{gupta2017optimal} and the revenue is compared between preemptive and non preemptive scheduling schemes.} We now give details of joint pricing and scheduling model below. 

\begin{figure}[h]\centering 
 \includegraphics[scale=0.4]{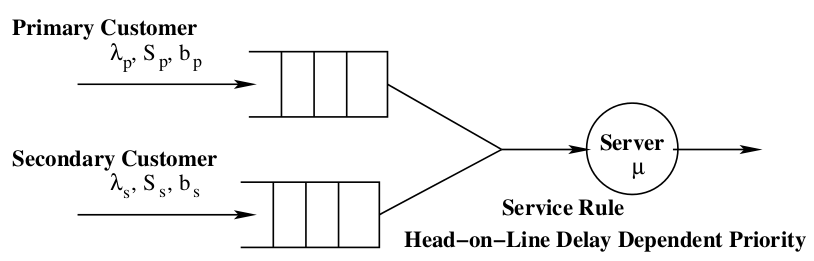}
 \caption{Schematic view of the model \cite{sinha2010pricing} }
 \label{fig:sks}
\end{figure}

\indent A schematic view of the model is shown in Figure \ref{fig:sks}. Primary class of customers arrive according to an independent Poisson arrival process with rate $\lambda_p$. $S_p$, the desired limit on the mean waiting time of the primary class of customers, indicates the service level offered. The service time of customers is independent and identically distributed with mean $1/\mu$ and variance $\sigma^2$, irrespective of customer class. Idea of the problem is to determine the promised limit on the mean waiting time of a secondary class of customers, $S_s$ and their unit admission price $\theta$ so as to maximize the revenue generated by the system, while constrained by primary class service levels.  The secondary class of customers arrive according to an independent Poisson arrival process with rate $\lambda_s$, which is dependent on $\theta$ and $S_s$:  $\lambda_s(\theta,S_s)=a - b\theta - cS_s$, where $a,~b,~c$ are positive constants driven by the market.

The mean waiting time of primary and secondary class customers depend on the queue scheduling rule. The scheduling discipline used in \cite{sinha2010pricing} was the non-preemptive delay dependent priority scheme, introduced by Kleinrock (see \cite{Kleinrock1964}). Let $\beta:=b_s/b_p$ is the delay dependent priority parameter. Note that $\beta=0$ corresponds to static high priority to primary class customers, $\beta=1$ is the global First Come First Serve (FCFS) queuing discipline across classes and $\beta = \infty$ corresponds to static high priority to secondary class customers. Let $W_{p}(\lambda_{s}, \beta)$ and $W_{s}(\lambda_{s}, \beta)$ be the mean waiting times of primary and secondary customers respectively, when the arrival rate of secondary jobs is $\lambda_{s}$ and queue management parameter is $\beta$.

Now select a suitable pair of pricing parameters $\theta$ and $S_{s}$ for the secondary class customers, a queue disciple management parameter $\beta$ and an appropriate admission rate for the secondary class customers $\lambda_{s}$, that will maximize the expected revenue from their inclusion, while ensuring that the mean waiting time to the primary class customers does not exceed a given quantity $S_p$. Thus, the revenue maximization problem, P0, is (see \cite{sinha2010pricing}): 
\begin{eqnarray}
\mbox{\textbf{P0:}\space}{\max_{\lambda_s,  \theta, S_s, \beta }~} \theta\lambda_s
\end{eqnarray}
\hspace{6cm} subject to:
\begin{eqnarray}
W_p(\lambda_s,\beta)\leq S_p \label{Pri_Qos}\\
W_s(\lambda_s,\beta)\leq S_s \label{Sec_Qos}\\
\lambda_s \leq \mu - \lambda_p \label{Sys_sta}\\
\lambda_s \leq a-b\theta-cS_s \label{Dem}\\
\lambda_s,\theta,S_s,\beta\geq 0 
\end{eqnarray}
Constraint (\ref{Pri_Qos}) and (\ref{Sec_Qos}) ensure the service level for primary and secondary class customers respectively. Constraint (\ref{Sys_sta}) is queue stability constraint. Constraint (\ref{Dem}) ensures that the mean arrival rate of secondary class customers should not exceed the demand generated
by charged price $\theta$ and offered service level $S_s$. This problem can be presented as following non-convex constrained optimization problem P1 (Constraints (\ref{Sec_Qos}) and (\ref{Dem}) are tight at optimality \cite{sinha2010pricing}):
\begin{eqnarray}
\mbox{\textbf{P1:}\space} \max_{\lambda_s,\beta}~ \dfrac{1}{b}\left( a\lambda_s -\lambda_s^2 -c \lambda_s W_s(\lambda_s,\beta)\right)  
\end{eqnarray}
\hspace{6cm} subject to:
\begin{eqnarray}
W_p(\lambda_s,\beta) \leq S_p \label{eqn:cons11}\\
\lambda_s \leq \mu - \lambda_p\\
\lambda_s ,\beta \geq  0 
\end{eqnarray}
Once the optimal secondary class mean arrival rate $\lambda_{s}^{*}$ and optimal queue discipline management parameter $\beta^{*}$ are calculated, the optimal admission price $\theta^{*}$ and optimal assured service level to secondary class $S_{s}^{*}$ can be computed using $S_{s}^{*} = W_{s}(\lambda_{s}^{*} ,\beta^{*})$ and $\theta^* = (a-\lambda_s^*-c S_s^*)/b$.

Note that above optimization problem P1 considers only finite values of $\beta$, though $\beta = \infty$ is also a valid decision variable as it corresponds to a static high priority to secondary class customers. Hence, solution of optimization problem P1 is obtained by decomposing the problem P1 in two parts (with finite and infinite $\beta$) and by comparing objectives (see \cite{sinha2010pricing}).

Optimization problem P1 can be transformed in following equivalent problem T1 by using the completeness and equivalence results between DDP and relative priority (see Theorem \ref{clm:RPequiv}). Problem T1 is comparatively easy to solve as optimization is over a compact set $p_1 \in[0,1]$ instead of $\beta \in [0,\infty]$. Thus, decomposition as in \cite{sinha2010pricing} is not needed while solving problem T1. 
\begin{eqnarray}
\mbox{\textbf{T1:}\space} \max_{\lambda_s,p}~ \dfrac{1}{b}\left( a\lambda_s -\lambda_s^2 -c \lambda_s W_s(\lambda_s,p)\right)  
\end{eqnarray}
\hspace{6cm} subject to:
\begin{eqnarray}
W_p(\lambda_s,p) \leq S_p \label{eqn:cons11}\\
\lambda_s \leq \mu - \lambda_p\\
\lambda_s \geq 0, 0 \leq  p \leq 1 
\end{eqnarray}
Note that the optimal solution to problem P1 or T1 is optimal over set of all non pre-emptive, {non-anticipative} and work conserving scheduling discipline from the virtue of completeness discussed in Section \ref{sec:completeness_proofs}.

\subsection{Optimal utility in data network}\label{utility_example}

{

The fundamental goal of any network design is to meet the needs of the users. An appropriate utility function often describes how the performance of an application depends on the delivered service. One can always increase the efficacy of an architecture by
deploying more bandwidth; faster speeds mean lower delays and fewer packet losses, and therefore higher utility values. Alternatively, for a given bandwidth (service rate), utility can be maximized by using optimal scheduling scheme.

 The utility maximization framework considered in \cite{jiang2002delay} is fairly generic. Jiang et. al. \cite{jiang2002delay} aim to maximize the utility in a delay sensitive data network by optimizing over probabilistic priority scheduling scheme. We now observe that their optimization problem is over a wider class (all $\pi \in \mathcal{F}$) by the completness of probabilistic priority in Theorem \ref{PP_completeness}. However, the mean waiting time expressions used in \cite{jiang2002delay} are approximate. Thus, the probabilistic priority parameter obtained by \cite{jiang2002delay} results in sub-optimal utility; we circumvent this problem by using relative priority scheme which is not only complete (see Theorem \ref{clm:RPcomplete}) but for which the closed form expressions for mean waiting times are also known. 


We first explain the utility framework of data network  considered in \cite{jiang2002delay}. Further, we obtain {optimal} utility by exploiting completeness of relative dynamic priority discussed in Section \ref{sec:completeness_proofs}. {Exact expressions enable us to study the impact of approximation on mean waiting time and optimal utility; we illustrate this by appropriate computational examples. }
}

Consider a network with single switch and two classes of customers. Delay experienced by a packet in the network can be approximated by sojourn time. The switch (e.g. Asynchronous transfer mode (ATM) switch) is modelled as a deterministic server with service time for a packet from either class as 1 unit time. Arrivals are according to independent Poisson processes with rate $\lambda_1$ and $\lambda_2$ for class 1 and class 2 respectively. Services provided by this network are differentiated into two classes: real-time service for real-time applications, and best-effort service for non-real-time applications. Without loss of generality, it is assumed that class 1 is for real-time service and class 2 is for best-effort service.

Real-time applications are usually delay sensitive. Such applications need their
data (packets) to arrive within a required delay. They perform badly if packets arrive later than this {required delay} bound. A pair $(d,b)$ is used to model this quality of service requirement. Here $d$ is the delay bound and $b$ is the acceptable probability that packets from real time class violate the delay bound.  Let $T_i^{\pi}$ be the sojourn time  experienced by class $i$, $i \in \{1,2\},$ under scheduling policy $\pi \in \mathcal{F}$ where $\mathcal{F}$ is set of all non pre-emptive, {non-anticipative} and work conserving scheduling disciplines. Let $v_1$ be the utility produced by the real time class if its quality of service requirement $(d,b)$ is met and $-v_2$ is the cost if the requirement is violated. Thus, utility function for real time class under scheduling policy $\pi$ is given by:
 \[
 u_1 =
  \begin{cases}
   v_1 & \text{if } P(T_1^\pi>d) \le b \\
   -v_2      & \text{if } P(T_1^\pi>d) > b
  \end{cases}
\]

On the other hand, non real-time applications do not have a delay bound requirement. Nevertheless, such applications usually prefer their data (packets) to be transmitted as quickly as possible. Let $v_3$ be the utility derived when packets from the best-effort class are transmitted infinitely fast and $v_4$ is the rate at which the utility declines as a function of the average sojourn time. Let $\bar{T}_i^\pi$ be the average sojourn time experienced by class $i$ packets under scheduling policy $\pi \in \mathcal{F}$. Thus, utility function for best effort class under scheduling policy $\pi$ is given by: 
$$u_2 = v_3 - v_4 \bar{T}_2^\pi$$
\[
\text{ Total utility } U:= u_1 +u_2 =
  \begin{cases}
   v_1 + v_3 - v_4 \bar{T}_2^\pi& \text{if } P(T_1^\pi>d) \le b \\
    v_3-v_2  - v_4 \bar{T}_2^\pi    & \text{if } P(T_1^\pi>d) > b
  \end{cases}
\]
Since the service time is deterministic 1 unit, thus 
$$T_i^\pi = W_i^\pi + 1,~~\bar{T}_i^\pi = \bar{W}_i^\pi + 1~\text{ for }i = 1, 2$$
Total utility under scheduling policy $\pi$ can be rewritten as:
\[
 U =
  \begin{cases}
   v_1 + v_3 - v_4 (1+\bar{W}_2^\pi) & \text{if } P(W_1^\pi>d-1) \le b \\
    v_3-v_2  - v_4 (1+\bar{W}_2^\pi)    & \text{if } P(W_1^\pi>d-1) > b
  \end{cases}
\]
By using tail probability approximation $P(W_i^\pi > x) \approx \rho e^{-\rho x / \bar{W}_i^\pi}$ from \cite{Wtime_approx}, total utility function, $U$, further simplifies to:
  \begin{numcases}{  U =}  
    v_1 + v_3 - v_4 (1+\bar{W}_2^\pi)   & $\bar{W}_1^\pi \le K$ \label{gfcfs1}\\
  v_3-v_2  - v_4 (1+\bar{W}_2^\pi)    & $\bar{W}_1^\pi > K$ \label{gfcfs2}
  \end{numcases}   
where $K :=  \dfrac{\rho(d-1)}{ln(\rho/b)}$. 

Then, one is interested in maximizing total utility function over all scheduling policy $\pi \in \mathcal{F}$ for given input parameters $v_1,\cdots,v_4$, arrival rates $\lambda_1,~\lambda_2$ and $(b,d)$ pair. It is easy to see that above utility is maximized by the scheduling policy for which $\bar{W}_1^\pi = K$. Note that $(\bar{W}_1^\pi, \bar{W}_2^\pi)$ satisfies following constraints for $\pi \in \mathcal{F}$. 
\begin{eqnarray}\label{Claw}
\rho_1 \bar{W}_1^\pi + \rho_2 \bar{W}_2^\pi = \frac{\rho}{1-\rho}W_0\\\label{W1bound}
\frac{W_0}{1-\rho_1} \le \bar{W}_1^\pi \le \frac{W_0}{(1-\rho)(1-\rho_2)}\\\label{W2bound}
\frac{W_0}{1-\rho_2} \le \bar{W}_2^\pi \le \frac{W_0}{(1-\rho)(1-\rho_1)}
\end{eqnarray}
Equation (\ref{Claw}) represents conservation law (see \cite{Kleinrock1965}). Equation (\ref{W1bound}) and (\ref{W2bound}) are the bounds on mean waiting times obtained by assigning strict priorities. For some values of $(d,b)$, $K$ can be beyond the above range of $\bar{W}_1^\pi$. In such cases, optimal utility for real time applications will be $v_1$ or $-v_2$ irrespective of scheduling policies. For either cases, system utility is maximized when $\bar{W}_2^\pi$ reaches its lower bound, i.e., strict priority is given to class 2. Hence, optimal scheduler produces the following system utility:
  \begin{numcases}{ U(OPT) =}
   v_1 + v_3  - v_4\left(1+\dfrac{W_0}{1-\rho_2}\right)   & $\hspace*{1.7cm} K >  \dfrac{W_0}{(1-\rho)(1-\rho_2)}$ \label{strictp1}\\
    v_3 + v_1  - v_4 \left[1 + \left(\dfrac{\rho \bar{W}_0}{1-\rho}- \rho_1K\right)/\rho_2   \right]    & $ \dfrac{W_0}{1-\rho_1} \le K \le \dfrac{W_0}{(1-\rho)(1-\rho_2)}$\label{puredynamic} \\
    v_3 - v_2 - v_4 \left(1+\dfrac{W_0}{1-\rho_2} \right)  & $\hspace*{1.7cm} K < \dfrac{W_0}{1-\rho_1}$\label{strictp2}
  \end{numcases}
Pure dynamic scheduling policy such that $\bar{W}_1^\pi =K$ is optimal for Equation (\ref{puredynamic}). It follows from completeness of dynamic priority schemes discussed in Section \ref{sec:completeness_proofs} that there exists a dynamic priority parameter $\pi$ for which $\bar{W}_1^\pi =K$ in a complete scheduling class. We consider relative priority and probabilistic priority schemes which are shown to be complete in Section \ref{sec:completeness_proofs}. Mean waiting times under relative priority scheme are known. Hence, we obtain the exact relative priority parameter to achieve optimal utility by the virtue of completeness in following theorem via $\bar{W_1}^\pi \equiv\bar{W_1}^{p^{RP}} =K$.

\begin{thm}\label{RP_theorem}
The maximum total utility over the  set of all non pre-emptive, {non-anticipative} and work conserving scheduling policies is achieved by implementing relative priority with the following parameter:
  \begin{numcases}{ p^{RP} =}
   \hspace*{4cm}0   & $\text{if } K >  \dfrac{W_0}{(1-\rho)(1-\rho_2)}$ \label{static1}\\
    \dfrac{\text{ln}(\frac{\rho}{b})W_0 - \rho(d-1)(1-\rho_2)(1-\rho)}{\rho\text{ln}(\frac{\rho}{b})W_0 + \rho(d-1)(\rho_2(1-\rho_2) - \rho_1(1-\rho_1))}    & $\text{if } \dfrac{W_0}{1-\rho_1} \le K \le \dfrac{W_0}{(1-\rho)(1-\rho_2)}$ \label{dynamic}\\
   \hspace*{4cm} 0  & $\text{if } K < \dfrac{W_0}{1-\rho_1}$\label{static2}
  \end{numcases}
\end{thm}
\begin{proof}
See Appendix \ref{proof:extra}.
\end{proof}

Now, consider probabilistic priority scheme which is also shown to be complete. Optimal probabilistic priority parameter can be obtained by solving $\bar{W_1}^\pi \equiv\bar{W_1}^{p^{PP}} =K$. Exact mean waiting time under probabilistic priority scheduling is not known. However, approximations are known (see \cite{jiang2002delay}). We now obtain the closed form expression for optimal approximate probabilistic priority parameter in following theorem to maximize utility via $\bar{W_1}^\pi \equiv\bar{W_1}^{p^{PP}_{approx}} =K$.


\begin{thm}\label{thm:pputility}
Approximate maximum total utility over the set of all non pre-emptive, {non-anticipative} and work conserving scheduling policies is achieved by implementing the following approximate probabilistic priority parameter:
\begin{numcases}{ p^{PP}_{approx}=}
	\hspace*{2cm}0   & $\text{if } K >  \dfrac{W_0}{(1-\rho)(1-\rho_2)}$ \label{PPstatic1}\\
    \dfrac{S^2-S(1+\rho_2)+\rho_2}{\rho_2-\rho S}    & $\text{if } \dfrac{W_0}{1-\rho_1} \le K \le \dfrac{W_0}{(1-\rho)(1-\rho_2)}$ \label{PPdynamic}\\
   \hspace*{2cm} 0  & $\text{if } K < \dfrac{W_0}{1-\rho_1}$\label{PPstatic2}
\end{numcases}
where $S = \dfrac{\rho(d-1)(1-\lambda_1)-W_0ln(\frac{\rho}{b})}{\rho(d-1) + (1-W_0)ln(\frac{\rho}{b})}$. 
\end{thm}
\begin{proof}
See Appendix \ref{proof:extra}.
\end{proof}

%
%
%
%

Note that the optimal scheduling policy is same under both relative and probabilistic priority for certain ranges of input parameters (see Equations (\ref{static1}), (\ref{static2}) and (\ref{PPstatic1}), (\ref{PPstatic2})). In such cases, $K = \dfrac{\rho (d-1)}{ln(\rho/b)}$ which depends on system input parameters, is beyond the range of $\bar{W}_1^\pi, ~\pi \in \mathcal{F}$, as in Equation (\ref{W1bound}). Thus, strict static priority to class 2 is optimal as discussed earlier and hence optimal scheduling policy is same under both relative and probabilistic priority scheduling for these ranges.

Approximate probabilistic priority parameter is obtained in Theorem \ref{thm:pputility} using approximate mean waiting time from \cite{jiang2002delay}. Thus, it is desirable to explore the  error in approximation. We first illustrate that approximate probabilistic priority parameter can be quite misleading; it can assign pure dynamic priority to a class when ({optimal}) relative priority is almost strict. Further, we illustrate the differences between {optimal} utility under relative priority scheduling and approximate utility under probabilistic priority scheduling. We calculate the approximate utility under global FCFS scheduling {scheme} due to its theoretical tractability.   

\subsubsection{Impact of $p^{PP}_{approx}$ on mean waiting times}

Consider the input parameters as $\lambda_1 = 0.25,~\lambda_2 = 0.25,~d = 4.912,~b = 0.01$ {(see experiment 9 in Table \ref{comparison_priority})}. $K$ turns out to be 0.5. At optimality, $\bar{W}_1 = K = 0.5$. Using conservation law, $\bar{W}_2 = 0.5$. Same mean waiting time for both classes with symmetric arrival rates will be achieved by global FCFS scheduling\footnote{ As demonstrated in Section \ref{Gfcfs}.}. Thus, the optimal priority parameter, whether it is relative or probabilistic priority, should be 0.5. Calculation of $p^{RP}$ for relative priority from Equation (\ref{dynamic}) indeed results in $p^{RP} = 0.5$. This verifies the exactness of mean waiting time expression of relative priority scheduling. While probabilistic priority from Equation (\ref{PPdynamic}) results in $p^{PP}_{approx} = 0.675$. This error is due to approximation in mean waiting time expression of probabilistic priority. 
\begin{table}[htb]
\centering
\begin{tabular}{|c|c|c|c|c|c|c|c|p{2.2cm}|p{2.62cm}|}
\hline 

 &\multicolumn{7}{c|}{ }   & \multicolumn{2}{|c|}{Scheduling schemes}\\
\hline
No. &$\lambda_1$ & $\lambda_2$ & $d$ & $b$ &$K$ &  $\bar{W}_1$ & $\bar{W}_2$ & Optimal relative priority ($p^{RP}$) & Approximate probabilistic priority ($p^{PP}_{approx}$)\\ 
\hline 
1 & 0.1182 & 0.26 & 4.912 & 0.01 & 0.4073 & 0.4073 & 0.2572 & 0.0159 & 0.5001 \\ 		
\hline 
2 & 0.37 & 0.1 & 4.912 & 0.01 & 0.4776 & 0.4776 & 0.3170 & 0.1712 & 0.5927 \\ 
\hline 
3 & 0.37 & 0.62 & 4.912 & 0.3 & 3.2438 & 3.2438 & 77.1045 & 0.9689 & 0.5754 \\ 
\hline 
4 & 0.47 & 0.15 & 4.912 & 0.01 & 0.5877 & 0.5877 & 1.5305 & 0.9954 & 0.9959 \\ 
\hline 
5 & 0.25 & 0.15 & 2.912 & 0.05 & 0.3678 & 0.3678 & 0.2759 & 0.2145 & 0.6413 \\ 
\hline 
6 & 0.23 & 0.15 & 2.912 & 0.05 & 0.3582 & 0.3582& 0.2270 & 0.0222 & 0.5807 \\ 
\hline 
7 & 0.3 & 0.2 & 3.3 & 0.0706 & 0.5875 & 0.5875 & 0.3668 & 0.1570 & 0.5001 \\ 
\hline 
8 & 0.4471 & 0.1 & 4.5 & 0.03 & 0.6595 & 0.6595 & 0.3558 & 0.1028 & 0.5000 \\
\hline  
9 & 0.25 & 0.25 & 4.912 & 0.01 & 0.5 & 0.5 & 0.5 & 0.5 & 0.675 \\
\hline  
\end{tabular} 
\caption{Optimal relative and approximate probabilistic priority parameters for various input instances}\label{comparison_priority}
\end{table}

Experiment 1, 7 and 8 in Table \ref{comparison_priority} show that optimal relative priority is close to 0 (static priority) while approximate probabilistic priority is close to 0.5 (global FCFS). Experiment 2 and 5 show the instances where relative priority results in higher priority to class 2 while probabilistic priority results in higher priority to class 1. Of course, there are some instances where approximation can be close to {optimal} results (see Experiment 4). In general, approximate parameters can be misleading (see Table \ref{comparison_priority}). 

\subsubsection{Impact of $p^{PP}_{approx}$ on optimal utility} 
Recall the problem of calculating the optimal scheduling parameter that maximizes the system utility as discussed in Section \ref{utility_example}. Given the system parameters $\lambda_1, \lambda_2, d, b, v_1, v_2, v_3, v_4$, one can find the optimal relative priority that achieves the maximum system utility as relative priority scheduling scheme is complete. Optimal relative priority parameter is given by Theorem \ref{RP_theorem} and 
{optimal} utility can be calculated using Equation (\ref{puredynamic}) as shown in Table \ref{compare_PP}.

Probabilistic priority scheme is also shown to be complete {in Section \ref{sec:completeness_proofs}} and hence one would be interested in calculating the optimal probabilistic priority parameter that maximizes the system utility. However, to do so, one needs to know the mean waiting times of both classes when a given probabilistic priority parameter is used, but, the only approximate mean waiting times are known (see \cite{jiang2002delay}). But, for global FCFS scheduling scheme  the mean waiting times of both classes are same as $\frac{W_0}{(1-\rho)}$ (see Section \ref{Gfcfs}). The system parameters in Table \ref{compare_PP} are chosen such that Theorem \ref{thm:pputility} results in $p^{PP}_{apprx}$ as $0.5$, which corresponds to global FCFS scheduling scheme. Put other way, for the system parameters as in Table 2, the available approximate mean waiting times for probabilistic priority scheme means that global FCFS should yield `maximal' utility. Using Equation (\ref{gfcfs1}) and (\ref{gfcfs2}), we calculate the utility obtained when global FCFS is used and we list them in the last column as `approximate utility'. { These approximations are in the computation of $p^{PP}_{approx}$ in Theorem \ref{thm:pputility}.} Note that, in these calculations, $K$ is dependent on system parameters and $\bar{W}_1^\pi = \bar{W}_1^{GFCFS} = \frac{W_0}{(1-\rho)}$. {Optimal} and approximate utilities are calculated for different instances of input parameters in Table \ref{compare_PP}. {It can be seen that approximate utility can be quite different from optimal utility (see Table \ref{compare_PP}); and can be misleading in some instances.}

\begin{table}[h]\centering
\begin{tabular}{|c|c|c|c|c|c|p{2.4cm}|c|p{2.8cm}|}
\hline 
$\lambda_1$ & $\lambda_2$ & $d$ & $b$ &  $v_3$  &$p^{RP}$ & {Optimal} Utility (using relative priority) & $p^{PP}_{approx}$ & Approx. Utility   (using probabilistic priority) \\ 
\hline 
 0.1179 & 0.26 & 4.911 & 0.01  &  300 & 0.0151 & 	209.16 & 0.5 & 203.55 \\ 
\hline 
0.301 & 0.1991 & 3.3 & 0.0706 & 300 &  0.1559 & 195.81 & 0.5 & 179.97 \\ 
\hline 
0.4471 & 0.1 & 4.5 & 0.03 &  300 &    0.1028 & 197.30 &  0.5 &  167.52 \\ 
\hline 
{0.16} & {0.382} & 6 & 0.01 & 500 &    0.3654 &  373.37  & 0.5 & 
{368.99} \\ 
\hline 
0.27 & 0.5284 & 4.9 & 0.1 &  600 &   0.6469 & 272.86   & 0.5 & 182.38 \\ 
\hline 
\end{tabular} 
\caption{{Optimal} vs approximate utility for different instances with $v_1 = v_2 = 60$ and $v_4 = 120$}\label{compare_PP}
\end{table}
\subsection{{Min-max fairness nature of global FCFS policy}}\label{Gfcfs}
{In this section, we introduce the notion of minmax fairness in multi-class queues and argue that a global FCFS policy is minmax fair by using the idea of completeness. Further, we find the explicit expressions of the weights given to extreme points to achieve global FCFS policy in 2-class queue.

We say that the {\em global FCFS} scheduling is employed in a multi-class queue if customers are served according to the order of their arrival times, irrespective of their class. The mean waiting time for each class is equal and given by $\frac{W_0}{(1-\rho)}$ in global FCFS policy. We now introduce the notion of certain minmax fairness and obtain priority parameters that achieve fairness among various classes in this sense.  }

In multi class queues, in addition to the focus on performance metrics such as waiting time, queue length, throughput etc., it is often important to ensure that the customers (jobs) are fairly treated. A vast literature has evolved in the refinement of the notion of fairness (see \cite{levy}, \cite{wierman_pe}, \cite{wierman_phd} and references therein). We introduce   another notion of fairness for multi-class queues: \textit{minimize the maximum dissatisfaction of each customer's class.} Here dissatisfaction is quantified in terms of the mean waiting time. Mathematically, it can be written as:
\begin{equation}
\min_{\pi \in \mathcal{F}}\max_{i \in \mathcal{I}}~(W_{i}^{\pi}) 
\end{equation}
where $\mathcal{I}$ is a finite set of classes and $\mathcal{F}$ is set of all work conserving, non pre-emptive and {non-anticipative} scheduling disciplines. Let $W_{i}^{\pi}$ be the mean waiting time for class $i ,~i \in \mathcal{I},$ customers when scheduling policy $\pi\in \mathcal{F}$ is employed. A minmax problem can also be described as an optimization problem via
lexicographic ordering (see \cite{osborne1994course}, \cite{vanam2013some}). We solve our minmax fairness problem by writing it as \textit{continuous semi-infinite program \footnote{A continuous optimization problem in a finite dimensional space with an uncountable set of constraints} }   (see, \cite{infi} for more details):
$$\hspace{-0.5in} \min_{\pi \in \mathcal{F}} \epsilon $$
\begin{eqnarray}
W_{i}^{\pi} & \leq &\epsilon, ~~\pi \in \mathcal{F}, i \in \mathcal{I} \\ 
\epsilon &\geq  & 0, 
\end{eqnarray}

{Consider a parametrized policy which is complete for $|\mathcal{I}|$ number of classes. Let the vector $\vec{\gamma}^\pi = \{\gamma_1^\pi, \gamma_2^\pi,... \gamma_{|\mathcal{I}|}^\pi\}$ be the parameter vector associated with this parameterized policy which determines a unique mean waiting time vector for $\pi \in \mathcal{F}$. The existence of such a complete parametrized policy is guaranteed from the synthesis algorithm of \cite{federgruen}, where a generalized delay dependent priority is used as a parametrized scheduling scheme. Thus, the above optimization problem can equivalently be solved by optimizing over the range of $\vec{\gamma}$}:
$$\hspace{-0.5in}\min_{\vec{\gamma}^\pi} \epsilon $$
\begin{eqnarray}\label{max_constraint}
W_{i}^{\vec{\gamma}^\pi} &\leq & \epsilon, ~~ i \in \mathcal{I}\\ 
\epsilon &\geq & 0,\\\label{conservation_law}
\sum_{i \in \mathcal{I}}\rho_i  W_{i}^{\vec{\gamma}^\pi} &=& \frac{\rho W_0}{(1-\rho)}
\end{eqnarray}
{Constraint (\ref{conservation_law}) is necessary as parametrized policy should satisfy the conservation law. Let {$W_{i}^{g},~i \in \mathcal{I},$} be the optimal solution of above optimization problem. We first argue that $W_{i}^{g},~i \in \mathcal{I},$ has to be equal for each class $i$ at optimality. It is clear from the conservation law (Equation (\ref{conservation_law})) that any deviation from equal mean waiting time policy will result in higher mean waiting time  (more than $W_{i}^{g}$) for at least one of the classes. Hence $\epsilon$ corresponding to that policy will be always more than $\epsilon$ corresponding to equal mean waiting time policy (due to Constraint (\ref{max_constraint})). Thus, the minima of the semi-infinite program will be given by $W_{i}^{g}$. Further, it follows from conservation law that $W_{i}^{g} = \frac{W_0}{(1-\rho)}$ for $i\in\mathcal{I}$. It will be attained by a suitable parameter as the class is complete. These parameters must implement global FCFS policy as each policy in the complete class corresponds to a mean waiting time vector (equal mean waiting times in this case). Thus, global FCFS policy is min-max fair. Note that global FCFS policy is realized by different parametrized dynamic priority policies discussed in this paper. Global FCFS is achieved by extended DDP, EDD, relative and HOL-PJ based priority by keeping all $b_i$'s, $u_i$'s, $p_i$'s and $D_i$'s equal respectively (see Equations (\ref{eqn:DDP_recursion}), (\ref{eqn:EDD_recursion}), (\ref{eqn:RP_recursion}) and (\ref{eqn:holpj_recursion}) respectively).

We now discuss a particular case of 2-class queues. In case of 2-class parametrized queueing system, global FCFS policy is realized by extended delay dependent priority with $\beta=1$, by EDD with $\bar{u} = 0$, by relative priority with $p_1 = p_2 = 1/2$ and by HOL-PJ dynamic priority with $\bar{D} = 0$. We find the weights given to the extreme points of line segment of Figure \ref{2classline} to achieve global FCFS in case of two classes. Consider weights $\alpha_1 = \frac{(1-\rho_1)}{(2-\rho_1 -\rho_2)}$ to class 1 and $\alpha_2 = \frac{(1-\rho_2)}{(2-\rho_1 -\rho_2)}$ to class 2. On simplifying, we have }
\begin{equation}
\begin{bmatrix}
\alpha_1 & \alpha_2
\end{bmatrix} \begin{bmatrix}
       W_{1}^{12} & W_{2}^{12}  \\
       W_{1}^{21} &  W_{2}^{21} \\
      
     \end{bmatrix} = \begin{bmatrix}
     \dfrac{W_0}{1-\rho}    & \dfrac{W_0}{1-\rho}   
        \end{bmatrix} 
\end{equation}
 Note that with two classes, we have exactly \textit{one and unique} pair of weights given to extreme points to get the global FCFS point in the interior of the polytope (line segment). Also note that, mean waiting time at this $\alpha_1$ and $\alpha_2$ is $\frac{W_0}{1-\rho}$ which is mean waiting time under global FCFS policy.

\section{Discussion}{{
The notion of completeness of scheduling schemes for mean waiting times vector is discussed for work conserving multi-class queueing systems. Four parametrized dynamic priorities (EDD, HOL-PJ, relative and PP) are shown to be complete for any 2-class M/G/1 queue. Equivalence between EDD, extended DDP, HOL-PJ and relative priority scheme is established. An explicit {nonlinear} one-to-one  transformation between the parameters of extended DDP and EDD policies (or relative priority) are obtained for mean waiting time vectors.

Significance of these results in optimal control of queueing systems is discussed. We formulate some relevant optimal control problems in contemporary areas such as high performance computing, cloud computing and characterize their optimal scheduling scheme. Further, an alternate but simple approach is devised for $c\mu$ rule and a joint pricing and scheduling problem. We obtain the {optimal} utility in 2-class data network by exploiting the completeness of relative priority discipline while approximate utility is obtained in literature by using approximate mean waiting times of probabilistic priority scheme. {A suitable notion of minmax fairness in multi-class queues is introduced and we note that the simple global FCFS scheme turns out to be minmax fair.}
%

It will be interesting to extend these ideas to $N$ class queues. Designing a new \textit{complete} dynamic priority scheme for given application domain can also be explored. The challenge would be to come up with a synthesis algorithm for this complete class, i.e., to devise an algorithm which computes the parameters of this complete class to achieve a given mean waiting time vector. } 
\label{sec:conclusion}
%
%


\appendix
\section{Extended delay dependent priority (DDP)}
\label{proof:DDP_cmplt}
{
Consider the extended delay dependent priority scheme where queue discipline management parameter need not be monotone. Thus, the queue discipline management parameter $b_i \ge 0$ for $i \in \{1,~2\}$. It is clear from Equation (\ref{eqn:DDP_recursion}) that average waiting time expressions in delay dependent priority (DDP) depends on ratios $b_i$ only. Define $\beta := b_2/b_1$. Average waiting time expression for class 1, $W_{1}^{\beta}$, and for class 2, $W_{2}^{\beta}$, in 2-class DDP can be derived using Equation (\ref{eqn:DDP_recursion}) as follows:

\begin{eqnarray}\label{eqn:DDPclass1}
W_{1}^{\beta} &=& \frac{W_0(1 -\rho(1-{\beta}))}{(1-\rho)(1 -\rho_1(1-{\beta}))}\mathbf{1}_{\{\beta \leq 1\}} 
+\frac{W_0}{(1-\rho)(1 -\rho_2(1-\frac{1}{\beta}))}\mathbf{1}_{\{\beta > 1\}}
\end{eqnarray}

\begin{eqnarray}\label{eqn:DDPclass2}
W_{2}^{\beta} &=&\frac{W_0}{(1-\rho)(1 -\rho_1(1-{\beta}))}\mathbf{1}_{\{\beta \leq 1\}}
+\frac{W_0(1 -\rho(1-\frac{1}{\beta})}{(1-\rho)(1 -\rho_2(1-\frac{1}{\beta}))}\mathbf{1}_{\{\beta > 1\}}
\end{eqnarray}
where $\rho_i = \frac{\lambda_i}{\mu_i}$, $\rho = \sum\limits_{i=1}^2\rho_i$, $W_0 = \sum\limits_{i=1}^2\frac{\lambda_i}{2}\left(\sigma_i^2 + \frac{1}{\mu_i^2} \right)$, $0 < \rho <  1$ and $\mathbf{1_{\{.\}}}$ is indicator function. Note that $\beta = 0$ and $\beta = \infty$ gives the corresponding mean waiting time when strict higher priority is given to class 1 and class 2 respectively. Also $\beta = 1$ gives the mean waiting time under global FCFS policy.
}

\section{Proofs of lemma and claims}
\label{proof:lemmaclaim}
{
\textbf{Remark:} Note that corresponding expressions ($\left.\beta = \frac{W_0 - (1-\rho_1)(1-\rho)\tilde{I}(\bar{u})}{W_0 + \rho_1(1-\rho)\tilde{I}(\bar{u})}\right. 
$, $\beta  = \frac{W_0 + \rho_2(1-\rho)I(\bar{u})}{W_0 - (1-\rho_2)(1-\rho)I(\bar{u})}
$, $\beta= \frac{(2 - \rho)(1 - p_1)}{1 - \rho(1 - p_1)} 
$, $\left.\beta = \frac{1 -\rho p_1}{(2 -\rho)p_1}\right.$) in all proofs are monotone in nature. Thus, it doesn't matter whether both end points are included or excluded in analysis. We avoid including both end points in all proofs for clarity.  
}

\begin{mylemma}{\ref{clm:equivalenceDDPnEDD}}
Global FCFS, strict priorities are given by $\bar{u} = 0, -\infty, \infty$ in EDD and $\beta = 1, 0, \infty$ in DDP respectively. This can be verified using the expression of mean waiting times with DDP and EDD dynamic priority (see Equation (\ref{eqn:DDPclass1}), (\ref{eqn:DDPclass2}) and (\ref{eqn:EDDcombined1}), (\ref{eqn:EDDcombined2})). On considering following two cases, we have
\begin{enumerate}
\item $-\infty \le \bar{u} < 0$ and $0 \le \beta < 1:$ On equating the mean waiting time for class 1 under these two dynamic priority using equations (\ref{eqn:DDPclass1}) and (\ref{eqn:EDDcombined1}):
\begin{equation}
 E(W) - \rho_2\int_0^{-\bar{u}}P(T_1(W)> y)dy = \frac{W_0(1 -\rho(1-{\beta}))}{(1-\rho)(1 -\rho_1(1-{\beta}))}
\end{equation}  
On simplifying the above equation for $\beta$, we have 
\begin{eqnarray}\label{eqn:relbetau1}
\beta = \frac{W_0 - (1-\rho_1)(1-\rho)\tilde{I}(\bar{u})}{W_0 + \rho_1(1-\rho)\tilde{I}(\bar{u})}
\end{eqnarray}
where $\tilde{I}(\bar{u}) = \int_0^{-\bar{u}}P(T_1(W)> y)dy$. As $\bar{u} \rightarrow 0, \tilde{I}(\bar{u}) \rightarrow 0$ so $\beta \rightarrow 1$ from above equation. Similarly, as $\bar{u} \rightarrow -\infty, \tilde{I}(\bar{u}) \rightarrow E(T_1(W))$ or $\dfrac{W_0}{(1-\rho)(1-\rho_1)}$. Hence $\beta \rightarrow 0$. So
$$-\infty \le \bar{u} < 0 \Leftrightarrow 0 \le \beta < 1$$ 
Above relation follows from Equation (\ref{eqn:relbetau1}) and by the fact that $\beta$ is monotonically increasing of $\bar{u}$ as $\tilde{I}(\bar{u})$ is monotonically decreasing. 
\item $0 \le \bar{u} \le \infty$ and $1 \le \beta \le \infty:$ Again on equating the mean waiting time for class 1 under these two dynamic priority using equations (\ref{eqn:DDPclass1}) and (\ref{eqn:EDDcombined1}):
\begin{equation}
E(W) + \rho_2\int_0^{\bar{u}}P(T_2(W)> y)dy = \frac{W_0}{(1-\rho)(1 -\rho_2(1-\frac{1}{\beta}))}
\end{equation} 
On simplifying the above equation for $\beta$, we have
\begin{equation}\label{eqn:relbetau2}
\beta  = \frac{W_0 + \rho_2(1-\rho)I(\bar{u})}{W_0 - (1-\rho_2)(1-\rho)I(\bar{u})}
\end{equation}
where $I(\bar{u}) = \int_0^{-\bar{u}}P(T_1(W)> y)dy$. As $\bar{u} \rightarrow 0,~ I(\bar{u}) \rightarrow 0$ so $\beta \rightarrow 1$ from above equation. Similarly, as $\bar{u} \rightarrow \infty, {I}(\bar{u}) \rightarrow E(T_2(W))$ or $\dfrac{W_0}{(1-\rho)(1-\rho_2)}$. Hence $\beta \rightarrow \infty$. So
$$0 \le \bar{u} \le \infty \Leftrightarrow 1 \le \beta \le \infty$$ 
Above relation follows from Equation (\ref{eqn:relbetau2}) and by the fact that $\beta$ is monotonically increasing function of $\bar{u}$ as ${I}(\bar{u})$ is monotonically increasing.
\end{enumerate}
\end{mylemma}
\begin{mylemma}{\ref{clm:RPequiv}}
Average waiting time expression in DDP depends on the value of $\beta$ from Equation (\ref{eqn:DDPclass1}). Hence, consider following two cases:

\begin{enumerate}
\item \textbf{\underline{ $0 \le \beta \leq 1$}} From Equation (\ref{eqn:DDPclass1}), average waiting time for class 1 is given by:
$$W_{1}|_{\beta \leq 1} = \frac{W_0(1 -\rho(1-{\beta}))}{(1-\rho)(1 -\rho_1(1-{\beta}))}$$
On simplifying the expression of mean waiting time under relative priority using Equation (\ref{eqn:2class_relative}): 
\begin{equation}\label{eqn:w1p1}
W_{1}|_{p = p_1 } = \frac{(1-\rho p_1)W_0}{(1-\rho)(1-\rho_2 - p_1(\rho_1-\rho_2))}
\end{equation} 
On simplifying the expressions for $W_{1}|_{p = p_1 } = W_{1}|_{\beta \leq 1} $, we have 
\begin{equation}
p_1 = \frac{1+(1 - \rho)(1-\beta)}{2-\rho(1-\beta)}
\end{equation}
Note that $\beta = 0 \rightarrow p_1 =1$ and $\beta = 1 \rightarrow p_1 =1/2$, on solving the above equation for $\beta$, we get
\begin{equation}
\beta= \dfrac{(2 - \rho)(1 - p_1)}{1 - \rho(1 - p_1)} 
\end{equation}
Note that, $\beta$ is a monotone function of $p_1$ and we are in the case of $0 \leq \beta \le 1$. 
$$0 \leq \beta \leq 1 \Leftrightarrow \frac{1}{2} \leq p_1 \leq 1$$ 

\item \textbf{\underline{$1 < \beta \leq \infty $}} From Equation (\ref{eqn:DDPclass1}), average waiting time for class 1 is given by:
$$W_{1}|_{\beta > 1} =\frac{W_0}{(1-\rho)(1 -\rho_2(1-\frac{1}{\beta}))} $$
 On equating above with Equation (\ref{eqn:w1p1}), we get
 \begin{equation}
 \beta = \dfrac{1 -\rho p_1}{(2 -\rho)p_1}
 \end{equation}
We are in case of $\beta > 1$ this gives $p_1 < 1/2$. Also note that $p_1 = 0  \rightarrow \beta = \infty$ and $\beta$ is a monotone function of $p_1$. Thus, we have $$ 1 < \beta \leq \infty \Leftrightarrow 0 \leq p_1 < \frac{1}{2}$$ 
\end{enumerate}
Hence, it follows that for every $\beta $ there exists $p$ and other way also. Similar arguments can be made if mean waiting time of class 2 is considered. Hence, result follows.
\end{mylemma}

\section{Proof of theorems}
\label{proof:extra}

\begin{mythm}{\ref{clm:EDDcomplete}}
Consider the notation $W_{i}^{12}$ to be the mean waiting time for class $i$, $i=1,2$,  when class 1 has strict priority over class 2. Mean waiting time are \cite{cobham}:
\begin{equation}
W_{1}^{12} = \dfrac{W_0}{1 - \rho_1}~~~~~\text{and}~~~~~~~W_{2}^{12} = \dfrac{1}{(1 - \rho_1)(1 - \rho_1 - \rho_2)}W_0
\label{eqn:class1p}
\end{equation}
Point $W_{12} = (W_{1}^{12}, W_{2}^{12})$  is shown in Figure \ref{2classline}. Similarly, when class 2 has strict priority over class~1, we get
\begin{equation}
W_{1}^{21} = \dfrac{W_0}{(1 - \rho_2)(1 - \rho_1 - \rho_2)}~~~\text{and}~~~W_{2}^{21} = \dfrac{1}{(1 - \rho_2)}W_0
\label{eqn:class2p}
\end{equation}
$W_{21} = (W_{1}^{21}, W_{2}^{21})$ is other extreme point shown in Figure \ref{2classline}. Consider the notation $W_i^{\alpha}$ for class $i$ as $W_i^{\alpha} = \alpha W_{i}^{12} + (1-\alpha)W_{i}^{21}$. On using Equation (\ref{eqn:class1p}) and (\ref{eqn:class2p}), we have
\begin{equation} \label{eqn:convexc}
W_1^{\alpha} = \dfrac{\alpha \rho_2 (\rho_1 + \rho_2 	- 2) + 1-\rho_1 }{(1 - \rho_1)(1 - \rho_2)(1 - \rho_1 - \rho_2)}W_0
\end{equation}
Average waiting time in EDD priority depends on $\bar{u}$ being positive or negative (see Equation (\ref{eqn:EDDcombined1}) and (\ref{eqn:EDDcombined2})), consider the following two cases:
\begin{enumerate}
\item $0 \leq \bar{u} \leq \infty:$ Expected waiting time for class 1 is given by (using Equation (\ref{eqn:EDDcombined1}))
\begin{eqnarray}\nonumber
E(W_1) &=& E(W) + \rho_2 \int_0^{\bar{u}}P(T_2(w)> y)dy\\
 &=& E(W) + \rho_2 I(\bar{u})\nonumber
\end{eqnarray}
On equating $E(W_1)$ with $W_1^{\alpha}$ and solving for $\alpha$, we have
\begin{equation}\label{eqn:alpha1}
\alpha = \dfrac{1-\rho_1}{2-\rho_1 -\rho_2} - \dfrac{I(\bar{u})(1-\rho_1)(1-\rho_2)(1-\rho)}{W_0(2-\rho_1 - \rho_2)}
\end{equation}
$\bar{u} = 0 \Rightarrow  I(\bar{u}) = 0$ so $\alpha = \dfrac{1-\rho_1}{2-\rho_1-\rho_2}$ and $\bar{u} = \infty \Rightarrow I(\bar{u}) = \int_0^{\infty}P(T_2(w)> y)dy = E(T_2(W)) = \dfrac{W_0}{(1-\rho)(1-\rho_2)}$ (See \cite[page 152]{EDDpriority}). On putting back this value of $I(\bar{u})$ in Equation (\ref{eqn:alpha1}), we get $\alpha = 0$. Since $I(\bar{u})$ is monotone increasing in $\bar{u}$, we have 
\begin{equation}
0 \leq \bar{u} \leq \infty \Leftrightarrow 0 \leq \alpha \leq \frac{1-\rho_1}{2-\rho_1 -\rho_2}
\end{equation}
\item $-\infty \leq \bar{u} < 0:$ Expected waiting time for class 1 is given by (using Equation (\ref{eqn:EDDcombined1}))
\begin{eqnarray} \nonumber
E(W_1) &=& E(W) - \rho_2 \int_0^{-\bar{u}}P(T_1(W)> y)dy\\\nonumber
&=& E(W) - \tilde{I}(\bar{u})
\end{eqnarray}
On equating $E(W_1)$ with $W_1^{\alpha}$ and solving for $\alpha$, we have 
\begin{equation}\label{eqn:alpha2}
\alpha = \dfrac{(1-\rho_1)(1-\rho_2)(1-\rho)\tilde{I}(\bar{u})}{(2-\rho_1 -\rho_2)W_0} + \dfrac{(1-\rho_1)}{(2-\rho_1 - \rho_2)}
\end{equation}
$\bar{u} \uparrow 0 \Rightarrow  \tilde{I}(\bar{u}) \downarrow 0$ so $\alpha \downarrow \dfrac{1-\rho_1}{2-\rho_1-\rho_2}$ and $\bar{u} = -\infty \Rightarrow \tilde{I}(\bar{u}) = \int_0^{\infty}P(T_1(w)> y)dy = E(T_1(W)) = \dfrac{W_0}{(1-\rho)(1-\rho_2)}$ (See \cite[page 152]{EDDpriority}). On putting back this value of $I(\bar{u})$ in Equation (\ref{eqn:alpha2}), we get $\alpha = 1$. Since $\tilde{I}(\bar{u})$ is monotonically decreasing in  $\bar{u}$, we have 
\begin{equation}
-\infty \leq \bar{u} < 0 \Leftrightarrow \frac{1-\rho_1}{2-\rho_1 -\rho_2} < \alpha \leq 1
\end{equation} 
\end{enumerate} 
Thus, entire range of $\alpha$ is achieved by unique value of $\bar{u}$. Similar arguments can be made if waiting time of class 2 is considered. Thus, the result follows.
\end{mythm}

\begin{mythm}{\ref{clm:RPcomplete}}
Equation (\ref{eqn:2class_relative}) gives the average waiting time of class 1 with relative priority:
\begin{equation}
W_{1}^{p_{1}} = \dfrac{(1-\rho p_1)}{(1-\rho_1 - p_2 \rho_2)(1-\rho_2 - p_1 \rho_1) - p_1 p_2 \rho_1 \rho_2} W_0
\end{equation}
On equating above equation with convex combination Equation (\ref{eqn:convexc}), we get following relation between $\alpha$ and $p_1$
\begin{equation}
p_1 = \dfrac{\alpha \rho_2 (2-\rho_1 - \rho_2)(1-\rho_2)(1-\rho_1 - \rho_2)}{(\alpha \rho_2(\rho_1+\rho_2-2)+ 1-\rho_1)(\rho_2(1-\rho_2) - \rho_1(1-\rho_1)) +\rho (1-\rho_1)(1-\rho_2)(1-\rho_1 - \rho_2))}
\label{eqn:p1-alpha}
\end{equation} 
$\alpha = 0 \Rightarrow p_1 = 0 $ and $\alpha = 1 \Rightarrow p_1 = 1 $  or $p_2 = 1-p_1 = 0$. $p_1$ is a monotone function of $\alpha$. Thus, for the entire range of $\alpha \in [0,1] ~\exists~ p_1$ in relative priorities defined by Equation (\ref{eqn:p1-alpha}). Similar arguments can be made if waiting time of class 2 is considered. Hence, result follows.

%
%
%
\end{mythm}
\begin{mythm}{\ref{RP_theorem}}
We exploit the completeness of relative dynamic priority from Theorem \ref{clm:RPcomplete} to obtain optimal scheduling parameter. It follows from this theorem that optimizing over set of all non pre-emptive, {\color{blue}non-anticipative} and work conserving scheduling policy is equivalent to optimize over relative priority scheduling. Strict priority to class 2 is optimal for case (\ref{strictp1}) and (\ref{strictp2}). Thus, $p^{RP} = 0$ will be the  optimal relative priority scheduler. Optimal policy for case (\ref{puredynamic}) can be obtained by solving  $\bar{W_1}^\pi \equiv\bar{W_1}^{p^{RP}} =K$. By using mean waiting time expression of relative priority from Equation (\ref{eqn:2class_relative}), we have
\begin{equation}
\dfrac{1-\rho p^{RP}}{(1- \rho_1 - (1-p^{RP}) \rho_2)(1 - \rho_2 - p^{RP}\rho_1)-p^{RP} (1-p^{RP}) \rho_1\rho_2}W_0= K = \frac{\rho(d-1)}{ln(\rho/b)}
\end{equation}   
where $W_0 = \sum\limits_{i=1}^2\lambda_i\bar{x}_i^2/2$. On simplifying for $p^{RP}$, we get
$$p^{RP} = \frac{\text{ln}(\frac{\rho}{b})W_0 - \rho(d-1)(1-\rho_2)(1-\rho)}{\rho\text{ln}(\frac{\rho}{b})W_0 + \rho(d-1)(\rho_2(1-\rho_2) - \rho_1(1-\rho_1))} $$ 
\end{mythm}

\begin{mythm}{\ref{thm:pputility}}
By exploiting the completeness of probabilistic priority from Theorem \ref{PP_completeness} and following the similar arguments as in the proof of Theorem \ref{RP_theorem}, strict priority to class 2 is optimal in Case (\ref{PPstatic1}) and (\ref{PPstatic2}). Thus, $p^{PP}_{approx}=0$ will be optimal probabilistic priority scheduler. Consider the approximate mean waiting time from approach 1 of class $i$, $i=1,2$ (see \cite{jiang2002delay}): 
\begin{equation}\label{approxPPwaiting}
\bar{W}_i = \frac{W_0 + s_{\bar{i}}\frac{(1-\beta_i)}{\beta_i}}{1-\rho_i-\lambda_i s_{\bar{i}}\frac{(1-\beta_i)}{\beta_i}}
\end{equation}
where $s_i$ is the first moment of service time for class $i$. Notation $\bar{i}$ denotes the class other than $i$, i.e., if $i=1,~2$, then $\bar{i}=2,1$ respectively and $\beta_i = (1-q_{\bar{i}}) + \omega_i q_{\bar{i}}$ where $q_i$'s are given by:
$$q_1 = \frac{[1+\omega_1\rho_1 - \omega_2\rho_2] - \sqrt{(1+\omega_1\rho_1-\omega_2\rho_2)^2-4\omega_1\rho_1}}{2\omega_1}$$
$$q_2 = \frac{[1+\omega_2\rho_2 - \omega_1\rho_1] - \sqrt{(1+\omega_2\rho_2-\omega_1\rho_1)^2-4\omega_2\rho_2}}{2\omega_2}$$
Note that $\omega_1 = p_1,$ $\omega_2 = 1-p_1$ and $\omega_i + \omega_{\bar{i}} = 1$ as discussed in Section \ref{2classPP}. Approximate optimal scheduling policy for Case (\ref{PPdynamic}) can be obtained by equating $$\bar{W_1}^\pi \equiv\bar{W_1}^{p^{PP}_{approx}} =K = \frac{\rho(d-1)}{ln(\rho/b)}$$
On using approximate mean waiting time from Equation (\ref{approxPPwaiting}) in above expression, we get approximate probabilistic priority parameter for Case (\ref{PPdynamic}) as given in theorem statement. 
\end{mythm}

\section{Optimal scheduling rule}\label{cmurule}
  $$\mathbf{T1}~~\min_{ p \in [0, 1] }~~c_1 W_{1}^{p}+c_2 W_{2}^{p} $$
Mean waiting time is given by Equation (\ref{eqn:2class_relative}) under relative priority scheduling scheme for two class queues. Define $f(p_1) \equiv c_1W_{1}^{p}+c_2 W_{2}^{p}$. On simplifying the objective using Equation (\ref{eqn:2class_relative}), we have 
\begin{equation}
f(p_1) = \frac{c_1 + c_2(1-\rho)- \rho(c_1 - c_2)p_1}{(1-\rho_1-(1-p_1)\rho_2)(1-\rho_2-p_1\rho_1)-p_1(1-p_1)\rho_1\rho_2}W_0
\end{equation}
On further simplifying, $f(p_1)$ can be re-written as:
$$f(p_1) = \frac{a_1 + a_2 p_1}{a_3 + a_4p_1}$$
where $a_1 = (c_1+c_2(1-\rho))W_0$, $a_2 = \rho(c_2-c_1)W_0$, $a_3 = (1-\rho)(1-\rho_2)$ and $a_4 = \rho_2(1-\rho_2) - \rho_1(1-\rho_1)$. Derivative of $f(p_1)$ with respect to $p_1$ is given by:
\begin{equation}\nonumber
f'(p_1) = \frac{(a_2 a_3 - a_1 a_4)}{(a_3 + a_4p_1)^2}
\end{equation}
Note that sign of derivative depends on sign of $(a_2 a_3 - a_1 a_4)$. Derivative will be positive if $a_2 a_3 > a_1 a_4$.  
$$a_2a_3 > a_1a_4 \Rightarrow \frac{c_2}{\rho_2} > \frac{c_1}{\rho_1} $$ 
Derivative is positive under above condition. Thus, objective function $f(p_1)$ will be increasing in $p_1$. $p_1 = 0$ will achieve the minimum objective value. This completes the proof of optimality of strict priority according to $c/\rho$ rule over all possible scheduling policies with two classes.

\end{document}